\documentclass[journal]{IEEEtran}
\IEEEoverridecommandlockouts
\usepackage{graphicx}
\usepackage{placeins}
\usepackage{float}
\usepackage{tabularx,colortbl}
\usepackage{amsthm}
\usepackage{amsmath}
\usepackage{mathtools}
\usepackage{amssymb}
\usepackage[named]{algo}
\usepackage[noend]{algpseudocode}
\usepackage{algorithm}
\usepackage{verbatim}
\usepackage{url}
\usepackage{cases}
\usepackage{amsmath}
\usepackage{amsfonts}
\usepackage{amssymb}
\usepackage{verbatim}
\usepackage{multirow}
\usepackage{graphicx}
\usepackage{epstopdf}
\usepackage{subfigure}
\usepackage{graphics}
\usepackage{color}
\newtheorem{lemma}{Lemma}

\newtheorem{theorem}{Theorem}

\usepackage{tikz}

\IEEEoverridecommandlockouts

\begin{document}
\title{SUMMeR: Sub-Nyquist MIMO Radar}
\author{David Cohen, Deborah Cohen, \emph{Student IEEE}, Yonina C. Eldar, \emph{Fellow IEEE} and Alexander M. Haimovich, \emph{Fellow IEEE} \thanks{
This project has received funding from the European Union's Horizon 2020 research and innovation program under grant agreement No. 646804-ERC-COG-BNYQ, and from the Israel Science Foundation under Grant no. 335/14. Deborah Cohen is grateful to the Azrieli Foundation for the award of an Azrieli Fellowship. }}
\maketitle


\maketitle

\begin{abstract}
Multiple input multiple output (MIMO) radar exhibits several advantages with respect to traditional radar array systems in terms of flexibility and performance. However, MIMO radar poses new challenges for both hardware design and digital processing. In particular, achieving high azimuth resolution requires a large number of transmit and receive antennas. In addition, the digital processing is performed on samples of the received signal, from each transmitter to each receiver, at its Nyquist rate, which can be prohibitively large when high resolution is needed. Overcoming the rate bottleneck, sub-Nyquist sampling methods have been proposed that break the link between radar signal bandwidth and sampling rate. In this work, we extend these methods to MIMO configurations and propose a sub-Nyquist MIMO radar (SUMMeR) system that performs both time and spatial compression. We present a range-azimuth-Doppler recovery algorithm from sub-Nyquist samples obtained from a reduced number of transmitters and receivers, that exploits the sparsity of the recovered targets' parameters. This allows us to achieve reduction in the number of deployed antennas and the number of samples per receiver, without degrading the time and spatial resolutions. Simulations illustrate the detection performance of SUMMeR for different compression levels and shows that both time and spatial resolution are preserved, with respect to classic Nyquist MIMO configurations. We also examine the impact of design parameters, such as antennas' locations and carrier frequencies, on the detection performance, and provide guidelines for their choice.

\end{abstract}

    \section{Introduction}

Multiple input multiple output (MIMO) \cite{fishler2004mimo} radar, which presents significant potential for advancing state-of-the-art modern radar in terms of flexibility and performance, poses new theoretical and practical challenges. This radar architecture combines multiple antenna elements both at the transmitter and receiver where each transmitter radiates a different waveform. 
Two main MIMO radar architectures are collocated MIMO \cite{li2007mimo} in which the elements are close to each other, and multistatic MIMO \cite{haimovich2008mimo} where they are widely separated. In this work, we focus on collocated MIMO.

Collocated MIMO radar systems exploit the waveform diversity, based on mutual orthogonality of the transmitted signals \cite{li2009mimo}. This generates a virtual array induced by the phase differences between transmit and receive antennas. Such systems thus achieve higher resolution than their phased-array counterparts with the same number of elements and transmissions, contributing to MIMO's popularity. This increased performance comes at the price of higher complexity in the transmitters and receivers design. MIMO radar systems belong to the family of array radars, which allow to recover simultaneously the targets range, Doppler and azimuth. This three-dimensional recovery results in high digital processing complexity.
One of the main challenges of MIMO radar is thus coping with complicated systems in terms of cost, high computational load and complex implementation.

Assuming a sparse target scene, where the ranges, Dopplers and azimuths lie on a predefined grid, the authors in \cite{strohmersparse, strohmersparse2} investigate compressed sensing (CS) \cite{CSBook} recovery for MIMO architectures. CS reconstruction is traditionally proposed to reduce the number of measurements required for the recovery of a sparse signal in some domain. However, in the works above, this framework is not used to reduce the spatial or time complexity, namely the number of antennas and samples, but is rather focused on mathematical guarantees of CS recovery in the presence of noise. To that end, the authors use a dictionary that accounts for every combination of azimuth, range and Doppler frequency on the grid and the targets' parameters are recovered by matching the received signal with dictionary atoms. The processing efficiency is thus penalized by a very large dictionary that contains every parameters combination. 

Several recent works have considered applying CS to MIMO radar to reduce the number of antennas or the number of samples per receiver without degrading resolution. The partial problem of azimuth recovery of targets all in the same range-Doppler bin is investigated in \cite{rossi2014spatial}. There, spatial compression is performed, where the number of antennas is reduced while preserving the azimuth resolution. Beamforming is applied on the time domain samples obtained from the thinned array at the Nyquist rate and the azimuths are recovered using CS techniques. In \cite{yu2010mimo, kalogerias2014matrix, mimoMC}, a time compression approach is adopted where the Nyquist samples are compressed in each antenna before being forwarded to the central unit. While \cite{yu2010mimo} exploits sparsity and uses CS recovery methods, \cite{kalogerias2014matrix, mimoMC} apply matrix completion techniques to recover the missing samples, prior to azimuth-Doppler \cite{kalogerias2014matrix} or range-azimuth-Doppler \cite{mimoMC} reconstruction. However, the authors do not address sampling and processing rate reduction since the compression is performed in the digital domain, after sampling, and the missing samples are reconstructed before recovering the targets parameters.    

In all the above works, the recovery is performed in the time domain on acquired or recovered Nyquist rate samples for each antenna. To reduce the sampling rate while preserving the range resolution, the authors in \cite{bar2014sub} consider frequency domain recovery. Similar ideas have been also used in the context of ultrasound imaging \cite{wagner2012compressed, chernyakova2014fourier}. The work in \cite{bar2014sub} demonstrates low-rate range-Doppler recovery in the context of radar with a single antenna, including sub-Nyquist acquisition and low-rate digital processing. Low-rate data acquisition is based on the ideas of Xampling \cite{mishali2011xampling, SamplingBook}, which consist of an ADC performing analog prefiltering of the signal before taking point-wise samples. Here, the samples are a sub-set of digitally transformed Fourier coefficients of the received signal, that contain the information needed to recover the desired signal parameters using CS algorithms \cite{CSBook}. A practical analog front-end implementing such a sampling scheme in the context of radar is presented in \cite{radar_demo}. To recover the targets range-Doppler from the sub-Nyquist samples, the authors introduce Doppler focusing, which is a coherent superposition of time shifted and modulated pulses. For any Doppler frequency, the received signals from different pulses are combined so that targets with corresponding Doppler frequencies come together in phase. This method improves the signal to noise ratio (SNR) by a factor of the number of pulses.

The work of \cite{bar2014sub} exploits the Xampling framework to break the link between radar signal bandwidth and sampling rate, which defines the time or range resolution. Here, we present the sub-Nyquist MIMO radar (SUMMeR) system, that extends this concept in the context of MIMO radar to break the link between the aperture and the number of antennas, which defines the spatial or azimuth resolution. We consider azimuth-range-Doppler recovery and apply the concept of Xampling both in space (antennas deployment) and in time (sampling scheme) in order to simultaneously reduce the required number of antennas and samples per receiver, without degrading the time and spatial resolution. In particular, we perform spatial and time compression while keeping the same resolution induced by Nyquist rate samples obtained from a full virtual array with low computational cost.

To this end, we express the ``Xamples", or compressed samples, both in time and space, in terms of the targets unknown parameters, namely range, azimuth and Doppler, and show how these can be recovered efficiently from the sub-Nyquist samples. We first focus on range-azimuth recovery and then extend our approach to range-azimuth-Doppler. In both cases, we present necessary conditions on the minimal number of samples and antennas for perfect recovery in noiseless settings. We then derive reconstruction algorithms by extending orthogonal matching pursuit (OMP) \cite{CSBook} for simultaneous sparse matrix recovery in order to solve a system of CS matrix equations. Our algorithm is inspired by matrix sketching \cite{cs_mat_yonina}, where only one matrix equation is considered. Our formulation is more complex as a result of coupling between the parameters. Besides, it involves simultaneous processing of several matrix equations, one per transmitter, to jointly recover the targets' range, azimuth and Doppler. We next show how our SUMMeR system can be enhanced so that the spatial compression does not decrease the detection performance.

An additional advantage of our approach concerns the Nyquist traditional regime both in time and space. In MIMO radar, two of the most popular approaches to ensure waveform orthogonality are code division multiple access (CDMA) and frequency division multiple access (FDMA). Although the narrowband assumption crucial for MIMO processing can hardly be applied to CDMA waveforms \cite{vaidyanathan2008mimo}, it is typically preferred. This is due to two essential drawbacks presented by FDMA: range-azimuth coupling \cite{lesturgie14, cattenoz2015mimo, rabaste2013signal} and limited range resolution to a single waveform's bandwidth \cite{vaidyanathan2009mimo, stralka11}. In this work, we adopt the FDMA framework and show that our processing overcomes these two drawbacks. This approach, as opposed to CDMA, allows to legitimately assume narrowband waveforms, which is key to azimuth resolution. This thus reconciles the trade-off between azimuth and range resolution. This topic is addressed in more details in \cite{cohen2016nyquist}.

The main contributions of our SUMMeR are as follows:
\begin{enumerate}
\item \textbf{Low rate sampling and digital processing} - the unknown targets parameters are recovered from sub-Nyquist samples obtained using Xampling. Both sampling and digital processing are performed at a low rate.
\item \textbf{Reduced number of antennas} - beamforming is performed on the Xamples obtained from a reduced number of transmit and receive antennas while keeping a fixed aperture.
\item \textbf{Scaling with problem size} - we separate the three dimensions (range, azimuth and Doppler) by adapting OMP to matrix form, with several matrix system equations. This avoids the use of a large CS dictionary, where each column corresponds to a range-azimuth-Doppler hypothesis.
\item \textbf{Maximal bandwidth exploitation} - the enhanced version of SUMMeR exploits the frequency bands left vacant by spatial compression for additional transmissions, increasing the detection performance while preserving the total bandwidth.
\item \textbf{Reconciled azimuth and range resolution trade-off} - FDMA waveforms simultaneously allow for narrowband single transmissions for high azimuth resolution and large total bandwidth for high range resolution.
\end{enumerate}

These properties are demonstrated in the simulations, which illustrate range-azimuth-Doppler recovery from low rate samples. In particular, we compare the detection performance of SUMMeR for different time and spatial compression levels with classic MIMO processing of Nyquist samples acquired on a full virtual array. We demonstrate that, under no compression, our FDMA processing achieves the detection performance of classic CDMA, even when wideband effects are neglected for the latter, giving it an advantage. We then demonstrate that our enhanced version of SUMMeR gains back the performance lost due to spatial compression. Last, we investigate the impact of several design parameters such as antennas' locations and transmissions' carrier frequencies on the detection performance and provide guidelines for their choice. 

    This paper is organized as follows. In Section~\ref{sec:classic}, we review the classic MIMO pulse-Doppler radar system and processing. The SUMMeR system is described in Section~\ref{sec:summer}. Section~\ref{sec:algo} introduces our sub-Nyquist sampling scheme and azimuth-range recovery algorithm, extended to range-azimuth-Doppler recovery in Section~\ref{sec:algo2}. Numerical experiments are presented in Section~\ref{sec:sim}.
   
\section{Classic MIMO Radar}
\label{sec:classic}

We begin by describing the classic MIMO radar architecture, in terms of array structure and waveforms, and the corresponding processing.

\subsection{MIMO Architecture}
The traditional approach to collocated MIMO adopts a virtual ULA structure \cite{chen2009signal}, where $R$ receivers, spaced by $\frac{\lambda }{2}$ and $T$ transmitters, spaced by $R\frac{\lambda }{2}$ (or vice versa), form two ULAs. Here, $\lambda$ is the signal wavelength.
Coherent processing of the resulting $TR$ channels generates a virtual array equivalent to a phased array with $TR$ $\frac{\lambda }{2}$-spaced receivers and normalized aperture $Z=\frac{TR}{2}$. This standard array structure and the corresponding virtual array are illustrated in Fig.~\ref{fig:arrays1} for $R=3$ and $T=5$. The blue circles represent the receivers and the red squares are the transmitters.

\begin{figure}[!h]
\begin{center}
\includegraphics[width=0.5\textwidth]{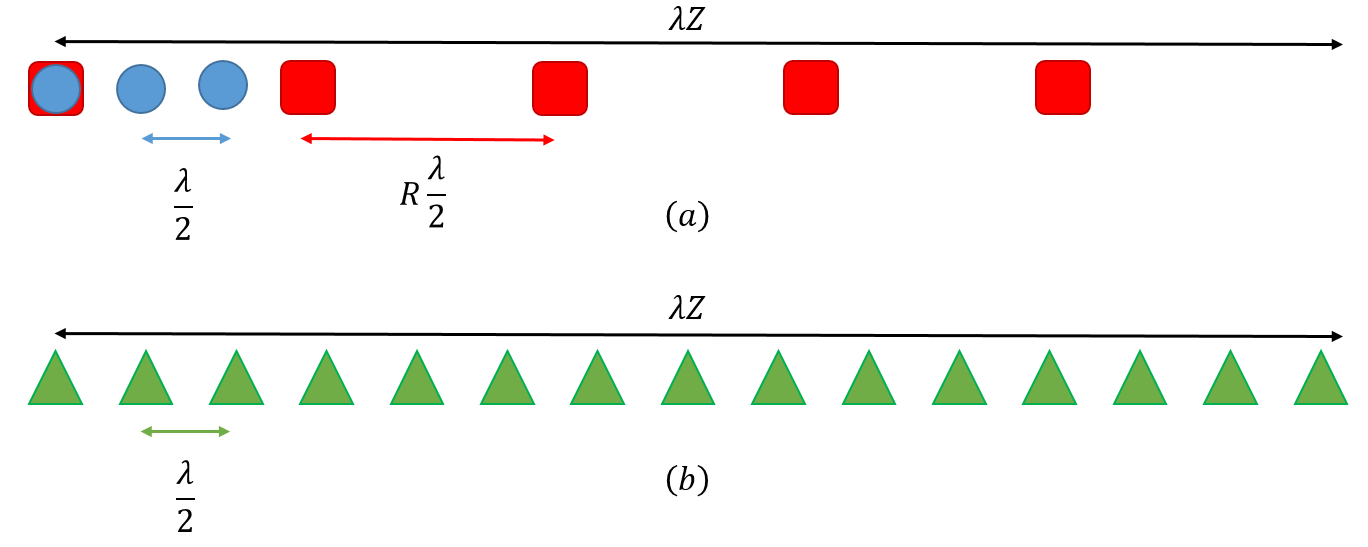}
\caption{Illustration of MIMO arrays: (a) standard array, (b) corresponding receiver virtual array.}
\label{fig:arrays1}
\end{center}
\end{figure}

Each transmitting antenna sends $P$ pulses, such that the $m$th transmitted signal is given by 
    \begin{equation}\label{trMth}
  s_m(t) = \sum_{p=0}^{P-1} {{{h}_{m}}\left( t -p\tau \right)}{{e}^{j2\pi {{f}_{c}}t}},\quad0\le t\le P\tau,
    \end{equation}
    where ${{h}_{m}}\left( t \right), 0\le m\le T-1 $ are narrowband and orthogonal pulses with bandwidth $B_h$, modulated with carrier frequency ${{f}_{c}}$. The coherent processing interval (CPI) is equal to $P \tau$, where $\tau$ denotes the pulse repetition interval (PRI). For convenience, we assume that $f_c \tau$ is an integer, so that the delay $e^{-j2\pi f_c \tau p}$ is canceled in the modulation for $0 \leq p \leq P-1$ \cite{peebles2007radar}. The pulse time support is denoted by $T_p$, with $0 < T_p < \tau$.

MIMO radar architectures impose several requirements on the transmitted waveform family. Besides traditional demands from radar waveforms such as low sidelobes, MIMO transmit antennas rely on orthogonal waveforms. In addition, to avoid cross talk between the $T$ signals and form $TR$ channels, the orthogonality condition should be invariant to time shifts, that is $\int_{-\infty }^{\infty }{{{s}_{i}}\left( t \right)s_{j}^{*}\left( t-\tau_0  \right)dt}=\delta \left( i-j \right),$ for $i,j \in \left[0,M-1\right]$ and for all $\tau_0$. This property implies that the orthogonal signals cannot overlap in frequency \cite{vaidyanathan2008mimo}, leading to FDMA. Alternatively, time invariant orthogonality can be approximately achieved using CDMA.

Both FDMA and CDMA follow the general model \cite{cattenoz2013}:
\begin{equation} \label{eq:gen}
h_m(t)=\sum_{u=1}^{N_c} w_{mu} e^{j2\pi f_{mu}t} v(t-u\delta_t),
\end{equation}
where each pulse is decomposed into $N_c$ time slots with duration $\delta_t$. 
Here, $v(t)$ denotes the elementary waveform, $w_{mu}$ represents the code and $f_{mu}$ the frequency for the $m$th transmission and $u$th time slot. The general expression (\ref{eq:gen}) allows to analyze at the same time different waveforms families. In particular, in CDMA, the orthogonality is achieved by the code $\{w_{mu} \}_{u=1}^{N_c}$ and $f_{mu}=0$ for all $1 \leq u \leq N_c$. In FDMA, $N_c=1$, $w_{mu}=1$ and $\delta_t =0$. The center frequencies $f_{mu}=f_m$ are chosen in $[ -\frac{TB_h}{2}, \frac{TB_h}{2} ]$ so that the intervals $[{{f}_{m}}-\frac{{{B}_{h}}}{2},{{f}_{m}}+\frac{{{B}_{h}}}{2}]$ do not overlap. For simplicity of notation, $\{h_m(t)\}_{m=0}^{T-1}$ can be considered as frequency-shifted versions of a low-pass pulse $v(t)=h_0(t)$ whose Fourier transform $H_0\left( \omega  \right)$ has bandwidth $B_h$, such that
    \begin{equation}
    {{H}_{m}}\left( \omega  \right)={{H}_{0}}\left( \omega -2\pi {{f}_{m}} \right).
    \end{equation}
We adopt a unified notation for the total bandwidth $B_{\text{tot}}=TB_h$ for FDMA and $B_{\text{tot}}=B_h$ for CDMA.

Consider $L$ non-fluctuating point-targets, according to the Swerling-0 model \cite{skolnik}. 
Each target is identified by its parameters: radar cross section (RCS) $\tilde{\alpha}_l$, distance between the target and the array origin or range $R_l$, velocity $v_l$ and azimuth angle relative to the array $\theta_l$. Our goal is to recover the targets' delay $\tau_l=\frac{2R_l}{c}$, azimuth sine $\vartheta_l=\sin (\theta_l)$ and Doppler shift $f_l^D = \frac{2v_l}{c}f_c$ from the received signals. In the sequel, the terms range and delay are interchangeable, as well as azimuth angle and sine, and velocity and Doppler frequency, respectively.


The following assumptions are adopted on the array structure and targets' location and motion, leading to a simplified expression for the received signal.
\begin{itemize}
\item[\textbf{A1}] Collocated array - target RCS $\tilde{\alpha}_l$ and $\theta_l$ are constant over the array (see \cite{haim2006} for more details)
\item[\textbf{A2}] Far targets - target-radar distance is large compared to the distance change during the CPI, which allows for constant $\tilde{\alpha}_l$,
\begin{equation}
v_l P \tau \ll \frac{c \tau_l}{2}.
\end{equation}
\item[\textbf{A3}] Slow targets - low target velocity allows for constant $\tau_l$ during the CPI,
\begin{equation} \label{eq:A3_1}
\frac{2 v_l P \tau}{c} \ll \frac{1}{B_{\text{tot}}},
\end{equation}
and constant Doppler phase during pulse time $T_p$,
\begin{equation} \label{eq:A3_2}
f_l^D T_p \ll 1.
\end{equation}
\item[\textbf{A4}] Low acceleration - target velocity $v_l$ remains approximately constant during the CPI, allowing for constant Doppler shift $f^D_l$,
\begin{equation}
\dot{v_l} P \tau \ll \frac{c}{2f_cP\tau}.
\end{equation}
\item[\textbf{A5}] Narrowband waveform - small aperture allows $\tau_l$ to be constant over the channels,  
\begin{equation} \label{eq:A1}
\frac{2Z \lambda}{c} \ll \frac{1}{B_{\text{tot}}}.
\end{equation}
\end{itemize}

\subsection{Received Signal}
The transmitted pulses are reflected by the targets and collected at the receive antennas. Under assumptions \textbf{A1}, \textbf{A2} and \textbf{A4}, the received signal $\tilde{x}_q(t)$ at the $q$th antenna is then a sum of time-delayed, scaled replica of the transmitted signals:
    \begin{equation}
    {\tilde{x}_{q}}\left( t \right) =  \sum\limits_{m=0}^{T-1} \sum\limits_{l=1}^{L} {\tilde{\alpha}_l {{s}_{m}}\left( \frac{c+v_l}{c-v_l} \left(t-\frac{R_{l,mq}}{c+v_l} \right)  \right)},
    \end{equation}
where $R_{l,mq}=2R_l-(R_{lm}+R_{lq})$, with $R_{lm}=\lambda \xi_m  \vartheta_l$ and $R_{lq}= \lambda \zeta_q \vartheta_l$ accounting for the array geometry, as illustrated in Fig.~\ref{fig:geom}.
\begin{figure}[!h]
\begin{center}
\includegraphics[width=0.4\textwidth]{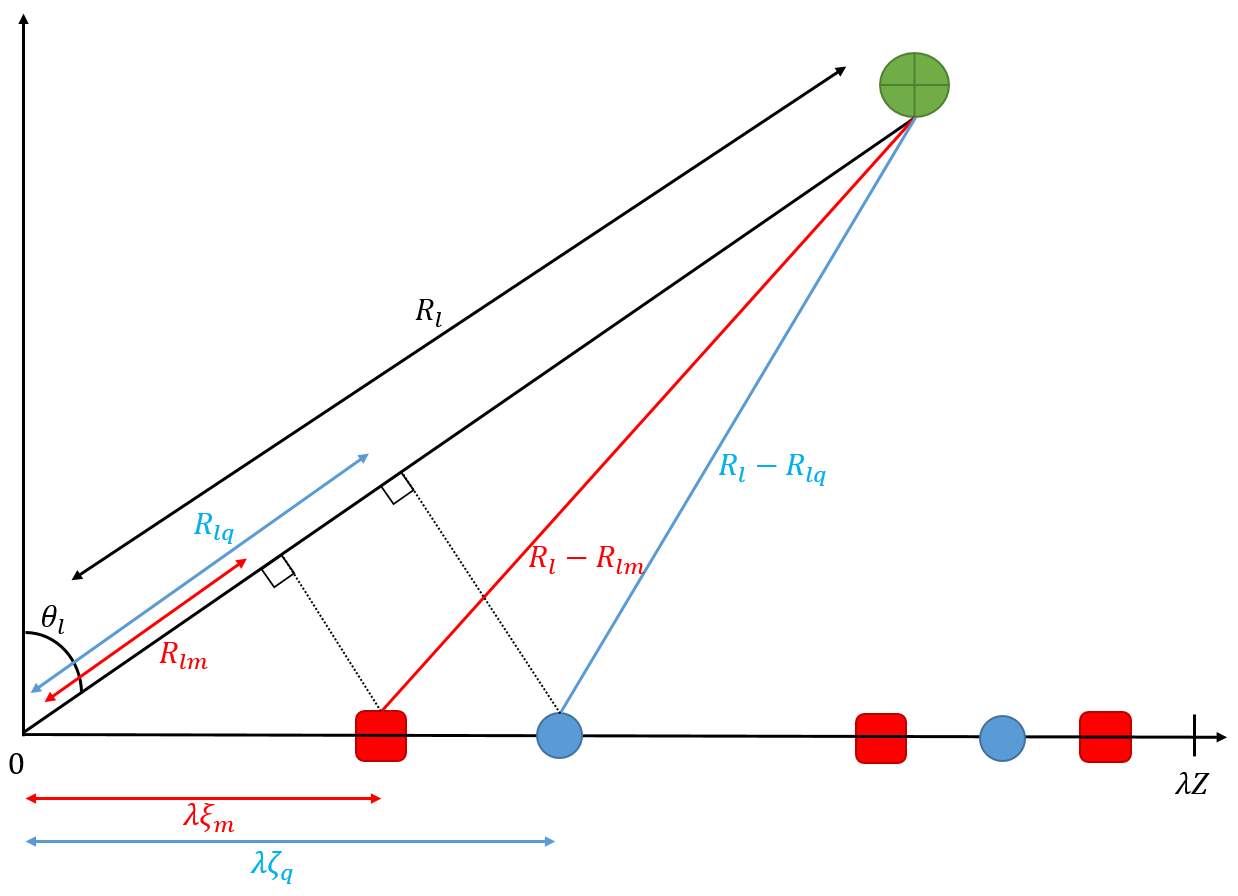}
\caption{MIMO array configuration.}
\label{fig:geom}
\end{center}
\end{figure}
The received signal expression is further simplified using the above assumptions. 
We start with the envelope $h_m(t)$ and consider the $p$th frame and the $l$th target. From $c \pm v_l \approx c$ and neglecting the term $\frac{2v_lt}{c}$ from \textbf{A3} (\ref{eq:A3_1}), we obtain
\begin{equation} \label{eq:simp1}
h_m \left( \frac{c+v_l}{c-v_l} \left(t-p\tau -\frac{R_{l,mq}}{c+v_l} \right) \right) = h_m(t-p\tau - \tau_{l,mq}).
\end{equation}
Here, $\tau_{l,mq}=\tau_l-\eta_{mq}\vartheta_l$ where $\tau_l=\frac{2R_l}{c}$ is the target delay and $\eta_{mq}=(\xi_m + \zeta_q) \frac{\lambda}{c}$ follows from the respective locations between transmitter and receiver.
We then add the modulation term of $s_m(t)$. Again using $c \pm v_l \approx c$, the remaining term is given by
\begin{equation} \label{eq:term2}
h_m(t-p\tau - \tau_{l,mq}) e^{j2\pi (f_c+f_l^D) \left(t-\tau_{l,mq} \right)}.
\end{equation}
After demodulation to baseband and using \textbf{A3} (\ref{eq:A3_2}), we further simplify (\ref{eq:term2}) to
\begin{equation} \label{eq:term3}
h_m(t-p\tau - \tau_{l,mq}) e^{-j2 \pi f_c \tau_l} e^{j2 \pi f_c \eta_{mq}\vartheta_l} e^{j2 \pi f^D_l p\tau}.
\end{equation}
The three phase terms in (\ref{eq:term3}) corresponds to the target delay, azimuth and Doppler frequency, respectively.
Last, from the narrowband assumption on $h_m(t)$ and \textbf{A5} (\ref{eq:A1}), the delay term $\eta_{mq}\vartheta_l$, that stems from the array geometry, is neglected in the envelope, which becomes
\begin{equation} \label{eq:term_1}
 h_m(t-p\tau - \tau_l).
\end{equation}
Substituting (\ref{eq:term_1}) in (\ref{eq:term3}), the received signal at the $q$th antenna after demodulation to baseband is given by
 \begin{equation} \label{eq:rec_sig1}
 x_q \left( t \right) = \sum\limits_{p=0}^{P-1} \sum\limits_{m=0}^{T-1} \sum\limits_{l=1}^{L}  \alpha_l h_m \left( t- p\tau  - \tau _{l} \right) e^{j2 \pi f_c \eta_{mq}\vartheta_l} e^{j2 \pi f^D_l p\tau},
\end{equation}
where $\alpha_l = \tilde{\alpha}_l e^{-j2 \pi f_c \tau_l}$.
In CDMA, the narrowband assumption on the waveforms $h_m(t)$ limits the total bandwidth $B_{\text{tot}}=B_h$, leading to a trade-off between time and spatial resolution \cite{vaidyanathan2008mimo}. In the next section, we show that in FDMA, this assumption can be relaxed with respect to the single bandwidth $B_h$, rather than $B_{\text{tot}}=TB_h$.

\subsection{Azimuth-Delay-Doppler Recovery}
Classic collocated MIMO radar processing traditionally includes the following stages:

\begin{enumerate}
  \item \textbf{Sampling:} at each receiver, the signal $x_q(t)$ is sampled at its Nyquist rate $B_{\text{tot}}$.
  \item \textbf{Matched filter:} the sampled signal is convolved with a sampled version of $h_m(t)$, for $0 \leq m \leq T-1$. The time resolution attained in this step is $1/B_h$.
  \item \textbf{Beamforming:} correlations between the observation vectors from the previous step and steering vectors corresponding to each azimuth on the grid defined by the array aperture are computed. The spatial resolution attained in this step is $2/TR$. In FDMA, this step leads to range-azimuth coupling, as illustrated in Section \ref{sec:sim}.
  \item \textbf{Doppler detection:} correlations between the resulting vectors and Doppler vectors, with Doppler frequencies lying on the grid defined by the number of pulses, are computed. The Doppler resolution is $1/P\tau$.
  \item \textbf{Peak detection:} a heuristic detection process is performed on the resulting range-azimuth-Doppler map. For example, the detection can follow a threshold approach \cite{richards2005fundamentals} or select the $L$ strongest points of the map, if the number of targets $L$ is known.
\end{enumerate}

In standard processing, the range resolution is thus governed by the signal bandwidth $B_h$. The azimuth resolution depends on the array aperture and is given by $\frac{2}{TR}$. Therefore, higher resolution in range and azimuth requires higher sampling rate and more antennas. The total number of samples to process, $NTRP$, where $N=\tau B_h$, can then become prohibitively high.
In order to break the link between time resolution and sampling rate on the one hand, and spatial resolution and number of antennas on the other hand, we propose to apply the Xampling framework \cite{mishali2011xampling} to both time (sampling scheme) and space (antennas deployment). 
Our goal can therefore be expressed more precisely as the estimation of the targets range, azimuth and velocities, i.e. ${{\tau}_{l}}$, ${{\vartheta }_{l}}$ and $f_l^D$ in (\ref{eq:rec_sig1}), while reducing the number of samples, transmit and receive antennas.

In this work, we adopt the FDMA approach, in order to exploit the narrowband property of the transmitted waveforms. Classic FDMA presents two main drawbacks. First, due to the linear relationship between the carrier frequency and the index of antenna element, a strong range-azimuth coupling occurs \cite{lesturgie14, cattenoz2015mimo, rabaste2013signal}. To resolve this aliasing issue, the authors in \cite{rias} use random carrier frequencies, which creates high sidelobe level. This can be mitigated by increasing the number of transmit antennas. The second drawback of FDMA is that the range resolution is limited to a single waveform's bandwidth, namely $B_h$, rather than the overall transmit bandwidth $B_{\text{tot}}=TB_h$ \cite{stralka11, vaidyanathan2009mimo}. Here, we overcome these two drawbacks. First, to resolve the coupling issue, we randomly distribute the antennas, while keeping the carrier frequencies on a grid with spacing $B_h$. Second, by coherently processing all the channels together, we achieve a range resolution of $B_{\text{tot}}=TB_h$. This way, we exploit the overall received bandwidth that governs the range resolution, while maintaining the narrowband assumption for each channel, which is key to the azimuth resolution. Our approach, presented in Sections \ref{sec:algo} and \ref{sec:algo2}, can be applied both in Nyquist and sub-Nyquist regimes, in time and space. Further discussion about FDMA versus CDMA can be found in \cite{cohen2016nyquist}.

\section{SUMMeR System Model}
\label{sec:summer}

The SUMMeR system implements compression in both space and time, reducing the number of antennas as well as the number of samples acquired by each receiver, while preserving range and azimuth resolution. We begin by describing the spatial compression. Time compression is introduced in Section \ref{sec:algo}.

Consider a collocated MIMO radar system with $M<T$ transmit antennas and $Q<R$ receive antennas, whose locations are chosen uniformly at random within the aperture of the virtual array described above, that is 
$\{\xi_{m}\}_{m=0}^{M-1} \sim \mathcal{U} \left[ 0, Z \right]$ and $\{\zeta _{q}\}_{q=0}^{Q-1} \sim \mathcal{U} \left[ 0,Z \right]$, respectively. Note that, in principle, the antenna locations can be chosen on the ULAs' grid. However, this configuration is less robust to range-azimuth ambiguity and leads to coupling between these parameters in the presence of noise, as shown in \cite{cohen2016nyquist}.
In Section~\ref{sec:algo}, we derive lower bounds on the number of antennas $M$ and $Q$. The spatially thinned array structure is illustrated in Fig.~\ref{fig:arrays2}, for $Q=2$ and $M=3$.

\begin{figure}[!h]
\begin{center}
\includegraphics[width=0.4\textwidth]{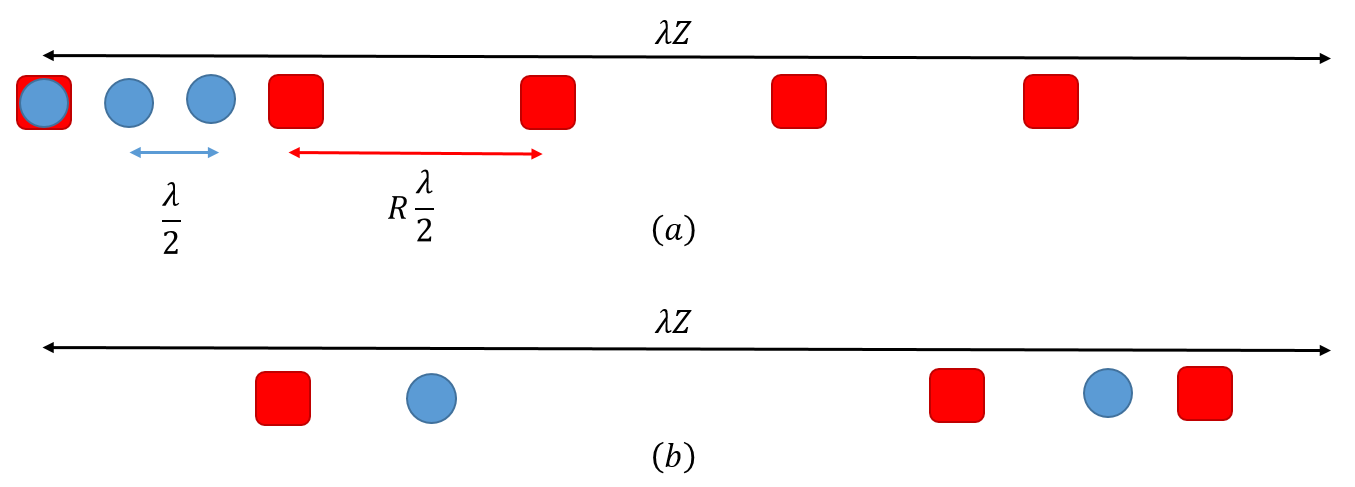}
\caption{Illustration of MIMO arrays: (a) standard array, (b) thinned array.}
\label{fig:arrays2}
\end{center}
\end{figure}

Since we adopt a FDMA framework, spatial compression, which in particular reduces the number of transmit antennas, removes the corresponding transmitting frequency bands as well. The transmitted signals are illustrated in Fig.~\ref{fig:specs} in the frequency domain. Figure~\ref{fig:specs}(a) and (b) show a standard FDMA transmission for $T=5$ and the resulting signal after spatial compression for $M=3$.

\begin{figure}[!h]
\begin{center}
\includegraphics[width=0.45\textwidth]{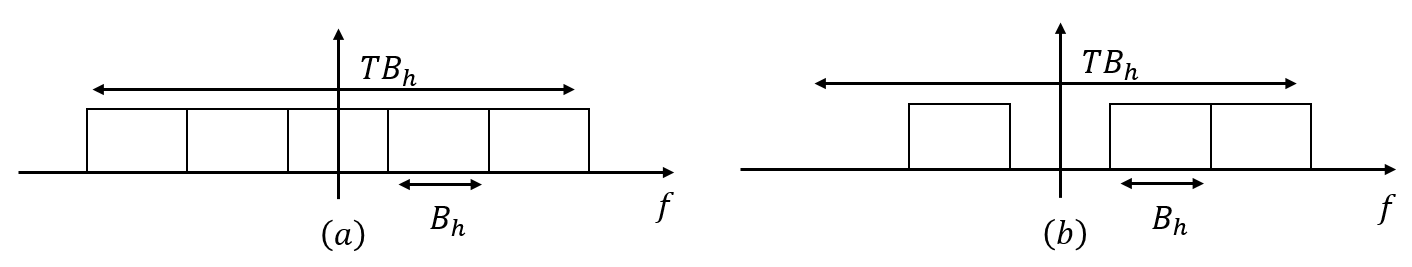}
\caption{FDMA transmissions: (a) standard, (b) spatial compression.}
\label{fig:specs}
\end{center}
\end{figure}

Our processing, described in Sections \ref{sec:algo} and \ref{sec:algo2}, allows to soften the strict neglect of the delay term in the transition from (\ref{eq:term3}) to (\ref{eq:term_1}). We only remove $\eta_{mq}\vartheta_l$ from the envelope $h_0(t)$, that stems from the array geometry. Then, (\ref{eq:term_1}) becomes
\begin{equation} \label{eq:term1_1}
 h_m(t-p\tau - \tau_l)e^{j2\pi f_m \eta_{mq} \vartheta_l}.
\end{equation}
Here, the restrictive assumption \textbf{A5} (\ref{eq:A1}) is relaxed to $\frac{2Z \lambda}{c} \ll \frac{1}{B_h}$. We recall that in CDMA, (\ref{eq:A1}) leads to a trade-off between azimuth and range resolution, by requiring either small aperture or small total bandwidth $B_{\text{tot}}$, respectively. Here, using the FDMA framework and less rigid approximation (\ref{eq:term1_1}), we need only the single bandwidth $B_h$ to be narrow, rather than the total bandwidth $B_{\text{tot}}$, eliminating the trade-off between range and azimuth resolution.
The received signal at the $q$th antenna after demodulation to baseband is in turn given by
 \begin{equation} \label{eq:rec_sig}
 x_q \left( t \right) = \sum\limits_{p=0}^{P-1} \sum\limits_{m=0}^{M-1} \sum\limits_{l=1}^{L}  \alpha_l h_m \left( t- p\tau  - \tau _{l} \right) e^{j2 \pi \beta_{mq} \vartheta _l} e^{j 2\pi f^D_l p \tau},
\end{equation}
where $\beta_{mq} = \left( \zeta _{q}+\xi _m \right) \left( f_m \frac{\lambda}{c} +1 \right)$.
It will be convenient to express $x_q(t)$ as a sum of single frames
\begin{equation}
\label{eq:frames}
x_q(t)= \sum_{p=0}^{P-1} x_q^p(t),
\end{equation}
where
\begin{equation}
\label{eq:one_frame}
x_q^p(t)= \sum_{m=0}^{M-1} \sum_{l=1}^{L} \alpha_l h(t-\tau_l - p\tau) e^{j2 \pi \beta_{mq} \vartheta _l} e^{j 2\pi f^D_l p \tau}.
\end{equation}

Our goal is to estimate the targets range, azimuth and velocity, i.e. to estimate ${{\tau}_{l}}$, ${{\vartheta}_{l}}$ and $f^D_l$ from low rate samples of $x_q(t)$, and a small number $M$ and $Q$ of antennas.

\section{Sub-Nyquist Range-Azimuth Recovery}
\label{sec:algo}

To introduce our sampling and processing, we begin by considering the special case of $P=1$, namely a unique pulse is transmitted by each transmit antenna. We describe how the range-azimuth map can be recovered from Xamples in time and space. We further derive necessary conditions on the number of channels and samples per receiver to allow for perfect range-azimuth recovery in noiseless settings. Subsequently, in Section \ref{sec:algo2}, we treat the general case where a train of $P>1$ pulses is transmitted by each antenna, and present a joint range-azimuth-Doppler recovery algorithm from Xamples, as well as recovery conditions.

\subsection{Xampling in Time and Space}
We begin by deriving an expression for the Fourier coefficients of the received signal, and show how the unknown parameters, namely $\tau_l$ and $\vartheta_l$ are embodied in these coefficients. We then briefly explain how the Fourier coefficients can be obtained from low rate samples of the signal.

The received signal $x_q\left( t \right)$ at the $q$th antenna is limited to $t\in \left[ 0,\tau \right]$ and thus can be represented by its Fourier series
    \begin{equation}
    \begin{array}{lll}
    x_q\left( t \right)=\sum\limits_{k\in \mathbb{Z}}{{{c}_{q}}\left[ k \right]{{e}^{-j2\pi kt/ \tau}}},\quad t\in \left[ 0,\tau \right],
    \end{array}
    \end{equation}
    where, for $-\frac{NT}{2} \leq k \leq \frac{NT}{2}-1$, with $N=\tau B_h$,
    \begin{eqnarray}
    \label{eq:coeffq}
   c_q\left[ k \right] &=& \frac{1}{\tau}\int_0^{\tau} x_q\left( t \right)e^{-j2\pi kt/ \tau}dt  \\
   & = &\frac{1}{\tau}\sum\limits_{m=0}^{M-1} \sum\limits_{l=1}^L \alpha_l e^{j2\pi \beta_{mq} \vartheta _l} e^{-j\frac{2\pi}{\tau}k \tau_l} H_m \left(\frac{2 \pi}{\tau}k \right). \nonumber
    \end{eqnarray}

In order to obtain the Fourier coefficients $c_{q}[k]$ in (\ref{eq:coeffq}) from low-rate samples of the received signal $x_q(t)$, we use the sub-Nyquist sampling scheme presented in \cite{bar2014sub, radar_demo}. For each received transmission, Xampling allows one to obtain an arbitrary set $\kappa$, comprised of $K = |\kappa|$ frequency components from $K$ point-wise samples of the received signal after appropriate analog preprocessing. Therefore, $MK$ Fourier coefficients are acquired at each receiver from $MK$ samples, with $K$ coefficients per frequency band or transmission.
    
Once the Fourier coefficients $c_{q}[k]$, for $k \in |\kappa|$, are acquired, we separate them into channels for each transmitter, by exploiting the fact that they do not overlap in frequency. Applying a matched filter, we have
   \begin{eqnarray}
     \tilde{c}_{q,m} \left[ k \right] &=&{{c}_{q}}\left[ k \right]H_{m}^{*}\left( \frac{2\pi }{\tau}k \right)  \\
    & = & \frac{1}{\tau} \left| H_m\left( \frac{2\pi }{\tau}k \right) \right|^2 \sum\limits_{l=1}^{L}\alpha_l e^{j2 \pi \beta_{mq} \vartheta_l} e^{-j\frac{2\pi}{\tau} k \tau_l}. \nonumber
    \end{eqnarray}
Let ${{y}_{m,q}}\left[ k \right]=\frac{\tau}{|H_0\left( \frac{2\pi }{\tau}k \right)|^2}{{\tilde{c}}_{q,m}}\left[ k+{{f}_{m}}\tau \right]$ be the normalized and aligned Fourier coefficients of the channel between the $m$th transmitter and $q$th receiver. Then,
    \begin{equation}
    \label{coffAlligned}
   y_{m,q}[k]=   \sum_{l=1}^{L} \alpha_l e^{j2\pi \beta_{mq} \vartheta_l}e^{-j\frac{2\pi }{\tau}k\tau_l} e^{-j2\pi f_m \tau_l},
       \end{equation}
       for $-\frac{N}{2} \leq k \leq \frac{N}{2}-1$. 
Our goal is then to recover the targets' parameters $\tau_l$ and $\theta_l$ from $ y_{m,q}[k]$.  
    
\subsection{Range-Azimuth Recovery Conditions}

We now derive the minimal number of channels $MQ$ and samples per receiver $MK$ required for perfect range-azimuth recovery from (\ref{coffAlligned}) in a noiseless environment. Theorem \ref{th:cont} considers continuous settings while in Theorem \ref{th:cond}, we assume that the delays and azimuths are confined to the Nyquist grid.

\begin{theorem}
\label{th:cont}
The minimal number of channels required for perfect recovery of $L$ targets in noiseless settings is $MQ \geq 2L$ with a minimal number of $MK \geq 2L$ samples per receiver.
\end{theorem}

\begin{proof}
Since there are no constraints on the delays or azimuths, let us examine the case where all targets have identical azimuth $\vartheta_l = \vartheta_0$. Equation (\ref{coffAlligned}) then becomes
 \begin{equation}
\label{eq:coeff_theta0}
   y_{m,q}[k]=   e^{j2\pi \beta_{mq} \vartheta_0} \sum_{l=1}^{L} \alpha_l e^{-j\frac{2\pi }{\tau}(k+f_m \tau) \tau_l}.
\end{equation}
For each channel, that is combination of the $m$th transmitter and $q$th receiver, we normalize (\ref{eq:coeff_theta0}) as $z_{m,q}[k]=y_{m,q}[k]/e^{j2\pi \beta_{mq} \vartheta_0}$, yielding
\begin{equation}
\label{eq:coeff_theta0n}
   z_{m,q}[k]=  \sum_{l=1}^{L} \alpha_l e^{-j\frac{2\pi }{\tau}(k+f_m \tau) \tau_l}.
\end{equation}
Since $f_m$ are chosen so that the frequency intervals $[{{f}_{m}}-\frac{{{B}_{h}}}{2},{{f}_{m}}+\frac{{{B}_{h}}}{2}]$ do not overlap, there are $MK$ distinct values of $k+f_m \tau$, for $k \in \kappa$ and $0 \leq m \leq M-1$.
If it were possible to solve (\ref{eq:coeff_theta0n}) with less than $MK$ samples, then we could use this to solve the one-dimensional problem of $2L$ delay-amplitude recovery with less than $2L$ samples, in contradiction with \cite{blu2002fir}. Therefore, it holds that $MK \geq 2L$.

Consider now the case where all targets have the same delay $\tau_l = \tau_0$. We obtain
 \begin{equation}
\label{eq:coeff_tau0}
   y_{m,q}[k]=   e^{-j\frac{2\pi }{\tau}k\tau_0} e^{-j2\pi f_m \tau_0} \sum_{l=1}^{L} \alpha_l e^{j2\pi \beta_{mq} \vartheta_l},
       \end{equation}
and after normalization,
\begin{equation}
\label{eq:coeff_tau0n}
  \tilde{z}_{m,q}[k] = \frac{1}{e^{-j\frac{2\pi }{\tau}(k+f_m \tau)\tau_0}} y_{m,q}[k]=  \sum_{l=1}^{L} \alpha_l e^{j2\pi \beta_{mq} \vartheta_l}.
       \end{equation}
Here, the number of distinct values of $\beta_{mq}$ is at most $MQ$. The upper bound can be achieved by an appropriate choice of $\xi_m$ and $\zeta_q$, for $0 \leq m \leq M-1$ and $0\leq q \leq Q-1$.
Applying the same considerations, we infer that $MQ \geq 2L$.
\end{proof}

As in traditional MIMO, suppose we now limit ourselves to the Nyquist grid with respect to the total bandwidth $TB_h$ so that $\tau_l = \frac{\tau}{TN}s_l$, where $s_l$ is an integer satisfying $0 \leq s_l \leq TN-1$ and $\vartheta_l = -1+\frac{2}{TR} r_l$, where $r_l$ is an integer in the range $0 \leq r_l \leq TR -1$.
Let $\mathbf{Y}^m$ be the $K \times Q$ matrix with $q$th column given by $y_{m,q}[k], k \in \kappa$. We can then write $\mathbf{Y}^m$ as
\begin{equation}
\label{eq:model}
\mathbf{Y}^m = \mathbf{A}^m \mathbf{X} \left(\mathbf{B}^m\right)^H.
\end{equation}
Here, $\mathbf{A}^m$ denotes the $K \times TN$ matrix whose $(k,n)$th element is $e^{- j \frac{2 \pi}{TN} \kappa_kn} e^{-j2\pi \frac{f_m}{B_h} \frac{n}{T}}$ with $\kappa_k$ the $k$th element in $\kappa$, $\mathbf{B}^m$ is the $Q \times TR$ matrix with $(q,p)$th element $e^{-j2 \pi \beta_{mq} (-1 +\frac{2}{TR}p)}$ and $(\cdot)^H$ denotes the Hermitian operator. The matrix $\bf X$ is a $TN \times TR$ sparse matrix that contains the values $\alpha_l$ at the $L$ indices $\left( s_l, r_l \right)$.

Our goal is to recover $\bf X$ from the measurement matrices $\mathbf{Y}^m, 0 \leq m \leq M-1$. The time and spatial resolution induced by $\bf X$ are $\frac{\tau}{TN}=\frac{1}{TB_h}$, and $\frac{2}{TR}$, respectively, as in classic CDMA processing. In traditional MIMO, $\bf X$ is recovered from Nyquist rate samples on a full virtual array, which is equivalent to full rank matrices $\bf A$ and $\bf B$, where 
\begin{equation} \label{eq:bigA}
\mathbf{A} = [ \mathbf{A}^{0^T} \, \mathbf{A}^{1^T} \, \cdots \, \mathbf{A}^{(M-1)^T}]^T,
\end{equation}
and 
\begin{equation} \label{eq:bigB}
\mathbf{B} = [ \mathbf{B}^{0^T} \, \mathbf{B}^{1^T} \, \cdots \, \mathbf{B}^{(M-1)^T} ]^T.
\end{equation}

To better grasp the structure of $\bf A$ and $\bf B$, consider the Nyquist regime, with carriers $f_m$ lying on the grid $f_m=m B_h$. In this case, the $(k,n)$th element of $\mathbf{A}^m$ is $e^{- j \frac{2 \pi}{TN} (k+mN)n}$ and $\bf A$ is the $TN \times TN$ Fourier matrix. Note that for a symmetric bandwidth, we choose $f_m = (m - \frac{T}{2})B_h$ and $\bf A$ becomes a Fourier matrix whose odd columns are multiplied by $-1$.
Similarly, assuming that the antenna elements lie on the virtual array illustrated in Fig.~\ref{fig:arrays1}, we have $\beta_{mq}=\frac{1}{2} (q + m R)$, where we used $f_m \frac{\lambda}{c} \ll 1$ to simplify the expression. Then, the $(q,p)$th element of $\mathbf{B}^m$ is $e^{- j \frac{2\pi}{TR} (q+mR)(p-\frac{TR}{2})}$ and $\bf B$ is the $TR \times TR$ Fourier matrix up to scaling and column permutation.
Here, due to the reduction in the number of antennas and samples per receiver, the number of rows of $\bf A$ and $\bf B$, namely $MK$ and $MQ$, respectively, is decreased. In terms of samples, reducing the number of transmitters decreases the number of measurement matrices $\mathbf{Y}^m$, reducing the number of receivers removes the corresponding columns of all matrices $\mathbf{Y}^m$ and reducing the number of samples per channels removes rows of $\mathbf{Y}^m$.

Theorem \ref{th:cond} presents necessary conditions on the minimal number of samples $MK$ and number of channels $MQ$ for perfect recovery of $\bf X$ from (\ref{eq:model}) under the grid assumption. 

\begin{theorem}
\label{th:cond}
The minimal number of channels required for perfect recovery of $\mathbf{X}$ with $L$ targets in noiseless settings is $MQ \geq 2L$ with a minimal number of $MK \geq 2L$ samples per receiver.
\end{theorem}
We note that we obtain the same recovery conditions as in the continuous case from Theorem \ref{th:cont}.
\begin{proof}
The observation model (\ref{eq:model}) for $0 \leq m \leq M-1$ can be equivalently written in vector form using the Kronecker product as
\begin{equation} \label{eq:Y_def}
\text{vec}(\mathbf{Y}) \triangleq
\left[ \begin{array}{c}
\text{vec}(\mathbf{Y}^0) \\ \text{vec}(\mathbf{Y}^1) \\ \vdots \\ \text{vec}(\mathbf{Y}^{M-1}) \end{array} \right]
= \left[ \begin{array}{c} \mathbf{\bar{B}}^0 \otimes \mathbf{A}^0 \\ \mathbf{\bar{B}}^1 \otimes \mathbf{A}^1 \\ \vdots \\ \mathbf{\bar{B}}^{M-1} \otimes \mathbf{A}^{M-1} \end{array} \right]
\text{vec}(\mathbf{X}).
\end{equation}
Here $\text{vec}(\mathbf{X})$ is a column vector that vectorizes the matrix $\bf X$ by stacking its columns, $\otimes$ denotes the Kronecker product and $\bf \bar{B}$ is the conjugate of $\bf B$. It follows that $\text{vec}(\mathbf{X})$ is $L$-sparse. Denote 
\begin{equation}
\label{eq:Cdef}
\mathbf{C} = \left[ 
\begin{array}{c} \mathbf{\bar{B}}^0 \otimes \mathbf{A}^0 \\ \mathbf{\bar{B}}^1 \otimes \mathbf{A}^1 \\ \vdots \\ \mathbf{\bar{B}}^{M-1} \otimes \mathbf{A}^{M-1}
\end{array} \right].
\end{equation}
In order to recover $\text{vec}(\mathbf{X})$ from $\text{vec}(\mathbf{Y})$, we require \cite{CSBook}
\begin{equation}
\label{eq:th_cond}
\text{spark} \left( \mathbf{C} \right) > 2L.
\end{equation}

We now state the following lemma whose proof is presented in Appendix \ref{sec:app}.

\begin{lemma}
\label{lemma:Cspark}
Let $\mathbf{A}^m \in \mathbb{C}^{K,N}$ and $\mathbf{B}^m \in \mathbb{C}^{Q,R}$, for $0 \leq m \leq M-1$ with $K \leq N$ and $Q \leq R$. Denote $\mathbf{A} = [ \mathbf{A}^{0^T} \, \mathbf{A}^{1^T} \, \cdots \, \mathbf{A}^{(M-1)^T} ]^T$ and $\mathbf{B} = [ \mathbf{B}^{0^T} \, \mathbf{B}^{1^T} \, \cdots \, \mathbf{B}^{(M-1)^T}]^T$.
Then,
\begin{equation}
\text{spark}(\mathbf{C}) = \min \{ \text{spark}(\mathbf{A}), \text{spark}(\mathbf{B}) \},
\end{equation}
where $\bf C$ is defined in (\ref{eq:Cdef}).
\end{lemma}
From Lemma \ref{lemma:Cspark}, (\ref{eq:th_cond}) holds iff
\begin{equation}
\label{eq:guarantees}
\text{spark} \left( \mathbf{A} \right)  > 2L  \quad \text{and} \quad \text{spark} \left( \mathbf{B} \right)  > 2L.
\end{equation}
Here $\bf A$ is of size $MK \times TN$ and $\bf B$ is of size $MQ \times TR$. This in turn leads to both $MK \geq 2L$ and $MQ \geq 2L$.
\end{proof}

Obviously, the design parameters $f_m, \xi_m, \zeta_q, \left| \kappa \right|$ should be chosen so that (\ref{eq:guarantees}) is satisfied. In the simulations, all of these parameters are first chosen at random. Then, deterministic guidelines for their choice are discussed in Section~\ref{sec:sim_param} and in \cite{cohen2016nyquist}.

\subsection{Range-Azimuth Recovery}
To recover the sparse matrix $\bf X$ from the set of equations (\ref{eq:model}), for all $0 \leq m \leq M-1$, where the targets' range and azimuth lie on the Nyquist grid, we would like to solve the following optimization problem
\begin{equation}
\label{eq:opt}
\text{min } ||\mathbf{X}||_0 \text{ s.t. } \mathbf{A}^m \mathbf{X} \left( \mathbf{B}^m\right)^H=\mathbf{Y}^m, \quad 0 \leq m \leq M-1.
\end{equation}
To this end, we extend the matrix OMP from \cite{cs_mat_yonina} to solve (\ref{eq:opt}), as shown in Algorithm \ref{algo:omp}. In the algorithm description, $\text{vec}(\mathbf{Y})$ is defined in (\ref{eq:Y_def}), $\mathbf{d}_t(l)= \left[ (\mathbf{d}_t^0(l))^T \, \cdots \, (\mathbf{d}_t^{M-1}(l))^T \right]^T$ where $\mathbf{d}_t^m(l) = \text{vec}(\mathbf{a}^m_{\Lambda_t(l,1)} ((\mathbf{\bar{b}}^m)^T_{\Lambda_t(l,2)})^T)$ with $\Lambda_t(l,i)$ the $(l,i)$th element in the index set $\Lambda_t$ at the $t$th iteration, and $\mathbf{D}_t= [\mathbf{d}_t(1) \, \dots \, \mathbf{d}_t(t)]$. Here, $\mathbf{a}^m_j$ denotes the $j$th column of the matrix $\mathbf{A}^m$ and it follows that $(\mathbf{b}^m)^T_j$ denotes the $j$th row of the matrix $\mathbf{B}^m$.
Once $\bf X$ is recovered, the delays and azimuths are estimated as
\begin{eqnarray}
\hat{\tau}_l &=& \frac{\tau}{TN}\Lambda_L(l,1), \label{eq:hat_tau} \\ \quad \hat{\vartheta}_l &=& -1+\frac{2}{TR}\Lambda_L(l,2). \label{eq:hat_theta}
\end{eqnarray}

\begin{algorithm}
\caption{OMP for simultaneous sparse matrix recovery}\label{algo:omp} 
\begin{algorithmic}[1]
\qinput Observation matrices $\mathbf{Y}^m$, measurement matrices $\mathbf{A}^m$, $\mathbf{B}^m$, for all $0 \leq m \leq M-1$
\qoutput Index set $\Lambda$ containing the locations of the non zero indices of $\mathbf{X}$, estimate for sparse matrix $\bf \hat{X}$
\State Initialization: residual $\mathbf{R}_0^m=\mathbf{Y}^m$, index set $\Lambda_0=\emptyset$, $t=1$
\State Project residual onto measurement matrices:
$$
\mathbf{\Psi} =\mathbf{A}^H \mathbf{R} \mathbf{B}
$$
where $\mathbf{A}$ and $\mathbf{B}$ are defined in (\ref{eq:bigA}) and (\ref{eq:bigB}), respectively, and $\mathbf{R} = \text{diag} \left( [ \mathbf{R}^0_{t-1} \, \cdots \, \mathbf{R}^{M-1}_{t-1}] \right)$ is block diagonal
\State Find the two indices $\lambda_t = [\lambda_t(1) \quad \lambda_t(2)]$ such that
$$
[\lambda_t(1) \quad \lambda_t(2)] = \text{ arg max}_{i,j} \left| \mathbf{\Psi}_{i,j} \right|
$$
\State Augment index set $\Lambda_t = \Lambda_t  \bigcup \{ \lambda_t \}$
\State Find the new signal estimate
$$
\mathbf{\hat{\alpha}}=[\hat{\alpha}_1 \, \dots \, \hat{\alpha}_t]^T = ( \mathbf{D}_t^T \mathbf{D}_t )^{-1} \mathbf{D}_t^T \text{vec}(\mathbf{Y}) 
$$
\State Compute new residual
$$
\mathbf{R}_t^m= \mathbf{Y}^m- \sum_{l=1}^t \alpha_l \mathbf{a}^m_{\Lambda_t(l,1)} \left( \mathbf{\bar{b}}^m_{\Lambda_t(l,2)} \right)^T$$
\State If $t < L$, increment $t$ and return to step 2, otherwise stop
\State Estimated support set $\hat{\Lambda}= \Lambda_L$ \\
Estimated matrix $\bf \hat{X}$: $\left( \Lambda_L(l,1), \Lambda_L(l,2) \right)$-th component of $\bf \hat{X}$ is given by $\hat{\alpha}_l$ for $l=1, \cdots, L$ while rest of the elements are zero
\end{algorithmic}
\end{algorithm}

Similarly, other CS recovery algorithms, such as FISTA \cite{FISTA_beck, FISTA_yang}, can be extended to our setting, namely to solve (\ref{eq:opt}).

\section{Sub-Nyquist Range-Azimuth-Doppler Recovery}
\label{sec:algo2}

In Section \ref{sec:algo}, we introduced Xampling in time and space for range-azimuth recovery. We now return to our original range-azimuth-Doppler recovery problem.  We begin by explaining how Xampling can be extended to the multi pulse signal (\ref{eq:rec_sig}). We then derive the minimal number of channels, samples per receiver and pulses per transmitter for perfect recovery in noiseless settings. Last, we present our range-azimuth-Doppler recovery algorithm based on the concept of Doppler focusing introduced in \cite{bar2014sub}.

Similarly to the derivations from Section \ref{sec:algo}, the $p$th frame of the received signal at the $q$th antenna, namely $x_q^p(t)$, is represented by its Fourier series
    \begin{equation}
    \begin{array}{lll}
    x_q^p(t)=\sum\limits_{k\in \mathbb{Z}}{{{c}_{q}^p}\left[ k \right]{{e}^{-j2\pi kt/ \tau}}},\quad t\in \left[ p\tau,(p+1)\tau \right],
    \end{array}
    \end{equation}
    where, for $-\frac{NT}{2} \leq k \leq \frac{NT}{2}-1$, with $N=\tau B_h$,
    \begin{equation}
    \label{eq:coeffq2}
    c_q^p\left[ k \right] = \frac{1}{\tau}\sum\limits_{m=0}^{M-1} \sum\limits_{l=1}^L \alpha_l e^{j2\pi \beta_{mq} \vartheta _l} e^{-j\frac{2\pi}{\tau}k \tau_l} e^{j 2 \pi f^D_l p \tau} H_m \left(\frac{2 \pi}{\tau}k \right).
    \end{equation}    
After separation to channels by matched filtering, the normalized and aligned Fourier coefficients ${{y}_{m,q}^p}\left[ k \right]=\frac{\tau}{|H_0\left( \frac{2\pi }{\tau}k \right)|^2}{{\tilde{c}}_{q,m}^p}\left[ k+{{f}_{m}}\tau \right]$, with ${{{\tilde{c}}}_{q,m}^p}\left[ k \right] = c_q^p\left[ k \right]H_{m}^{*}\left( \frac{2\pi }{\tau}k \right)$, are given by

    \begin{equation}
    \label{coffAlligned2}
   y_{m,q}^p[k]=   \sum_{l=1}^{L} \alpha_l e^{j2\pi \beta_{mq} \vartheta_l}e^{-j\frac{2\pi }{\tau}k\tau_l} e^{-j2\pi f_m \tau_l} e^{j 2 \pi f^D_l p \tau},
       \end{equation}
       for $-\frac{N}{2} \leq k \leq \frac{N}{2}-1$. 
The Fourier coefficients $y_{m,q}^p[k]$ of the frames of each channel (\ref{coffAlligned2}) are identical to (\ref{coffAlligned}) except for the additional Doppler term $e^{j 2 \pi f^D_l p \tau}$.

\subsection{Range-Azimuth-Doppler Recovery Conditions}

Theorems \ref{th:cont2} and \ref{th:cond2} consider the minimal number of channels and samples per receiver required for perfect range-azimuth-Doppler recovery. Again, we consider both continuous and discrete settings, where in the latter, we assume that the delays, azimuths and Doppler frequencies lie on the Nyquist grid.

\begin{theorem}
\label{th:cont2}
The minimal number of channels required for perfect recovery of $L$ targets in noiseless settings is $MQ \geq 2L$ with a minimal number of $MK \geq 2L$ samples per receiver and $P \geq 2L$ pulses per transmitter.
\end{theorem}

\begin{proof}
Repeating the proof of Theorem \ref{th:cont} for the cases where all targets have the same azimuth and Doppler frequency, we infer that $MK \geq 2L$. Similarly, the case where they all have the same delay and Doppler frequency yields $MQ \geq 2L$.
Consider now the situation where all targets have the same delay $\tau_l = \tau_0$ and azimuth $\vartheta_l=\vartheta_0$. Then,
 \begin{equation}
\label{eq:coeff_f0}
   y_{m,q}^p[k]=   e^{-j\frac{2\pi }{\tau_0}k\tau_l} e^{-j2\pi f_m \tau_0} e^{j2\pi \beta_{mq} \vartheta_0}\sum_{l=1}^{L} \alpha_l e^{j 2 \pi f^D_l p \tau}.
       \end{equation}
For each channel, we normalize (\ref{eq:coeff_f0}) as 
\begin{equation}
\label{eq:coeff_theta0n2}
   z_{m,q}^p[k]=  \frac{y_{m,q}^p[k]}{e^{-j\frac{2\pi }{\tau_0}k\tau_l} e^{-j2\pi f_m \tau_0} e^{j2\pi \beta_{mq} \vartheta_0}} = \sum_{l=1}^{L} \alpha_l e^{j 2 \pi f^D_l p \tau}.
\end{equation}
Applying the same considerations as in the proof of Theorem \ref{th:cont}, we conclude that $P \geq 2L$.
\end{proof}

As in Section \ref{sec:algo}, we next assume that the time delays, azimuths and Doppler frequencies are aligned to a grid. In particular, $\tau_l = \frac{\tau}{TN}s_l$, $\vartheta_l = -1+\frac{2}{TR} r_l$ and $f^D_l = -\frac{1}{2\tau}+\frac{1}{P\tau}u_l$, where $s_l, r_l$ and $u_l$ are integers satisfying $0 \leq s_l \leq TN-1$, $0 \leq r_l \leq TR -1$ and $0 \leq u_l \leq P-1$, respectively.
Let $\mathbf{Z}^m$ be the $KQ \times P$ matrix with $q$th column given by the vertical concatenation of $y_{m,q}^p[k], k \in \kappa$, for $0 \leq q \leq Q-1$. We can then write $\mathbf{Z}^m$ as
\begin{equation}
\label{eq:model2}
\mathbf{Z}^m = \left( \mathbf{\bar{B}}^m \otimes  \mathbf{A}^m \right) \mathbf{X}_D \mathbf{F}^H.
\end{equation}
Here, the $K \times TN$ matrix $\mathbf{A}^m$ and the $Q \times TR$ matrix $\mathbf{B}^m$ are defined in Section \ref{sec:algo} and $\bf F$ denotes the $ P \times P$ Fourier matrix.
The matrix $\mathbf{X}_D$ is a $T^2NR \times P$ sparse matrix that contains the values $\alpha_l$ at the $L$ indices $\left( r_l TN +s_l, u_l \right)$.

Our goal is now to recover $\mathbf{X}_D$ from the measurement matrices $\mathbf{Z}^m, 0 \leq m \leq M-1$. The time, spatial and frequency resolution stipulated by $\mathbf{X}_D$ are $\frac{1}{TB_h}$, $\frac{2}{TR}$ and $\frac{1}{P \tau}$ respectively. 

Theorem \ref{th:cond2} presents necessary conditions on the minimal number of channels $MQ$, samples per receiver $MK$ and pulses per transmitter $P$ for perfect recovery of $\mathbf{X}_D$ from (\ref{eq:model2}) under the grid assumption. 

\begin{theorem}
\label{th:cond2}
The minimal number of channels required for perfect recovery of $\mathbf{X}_D$ with $L$ targets in noiseless settings is $MQ \geq 2L$ with a minimal number of $MK \geq 2L$ samples per receiver and $P \geq 2L$ pulses per transmitter.
\end{theorem}

\begin{proof}
The observation model (\ref{eq:model2}) can be equivalently written as
\begin{equation}
\label{eq:model2F}
\mathbf{Z}^m \mathbf{F} = P \left( \mathbf{\bar{B}}^m \otimes  \mathbf{A}^m \right) \mathbf{X}_D,
\end{equation}
or in vector form,
\begin{equation}
\text{vec}(\mathbf{ZF}) \triangleq
\left[ \begin{array}{c}
\text{vec}(\mathbf{Z}^0\mathbf{F}) \\ \text{vec}(\mathbf{Z}^1\mathbf{F}) \\ \vdots \\ \text{vec}(\mathbf{Z}^{M-1}\mathbf{F}) \end{array} \right]
= P \mathbf{C}_D \text{vec}(\mathbf{X}_D),
\end{equation}
where
\begin{equation}
\label{eq:Cdef2}
\mathbf{C}_D =
\left[ \begin{array}{c} \mathbf{I}_P \otimes \mathbf{\bar{B}}^0 \otimes \mathbf{A}^0 \\ \mathbf{I}_P \otimes \mathbf{\bar{B}}^1 \otimes \mathbf{A}^1 \\ \vdots \\ \mathbf{I}_P \otimes \mathbf{\bar{B}}^{M-1} \otimes \mathbf{A}^{M-1} \end{array} \right],
\end{equation}
and $\mathbf{I}_P$ denotes the $P \times P$ identity matrix.
In order to recover the $L$-sparse vector $\text{vec}(\mathbf{X}_D)$ from $\text{vec}(\mathbf{Z})$, we require  $\text{spark} \left( \mathbf{C}_D \right) > 2L$ \cite{CSBook}.

Applying Lemma \ref{lemma:Cspark} twice, we obtain
\begin{equation}
\text{spark}(\mathbf{C}_D) = \min \{ \text{spark}(\mathbf{A}), \text{spark}(\mathbf{B}),  \text{spark}(\mathbf{\tilde{I}}_P) \},
\end{equation}
where $\mathbf{\tilde{I}}_P$ is the $MP \times P$ matrix which vertically concatenates $M$ times the matrix $\mathbf{I}_P$. Obviously, $\text{spark}(\mathbf{\tilde{I}}_P) = \text{spark}(\mathbf{I}_P)$.
Therefore, (\ref{eq:th_cond}) holds iff
\begin{equation}
\label{eq:guarantees2}
\text{spark} \left( \mathbf{A} \right)  > 2L, \quad \text{spark} \left( \mathbf{B} \right)  > 2L \quad \text{and} \quad \text{spark} \left( \mathbf{I}_P \right)  > 2L,
\end{equation}
which in turn leads to $MK \geq 2L$, $MQ \geq 2L$ and $P \geq 2L$.
\end{proof}

\subsection{Range-Azimuth-Doppler Recovery}

To recover jointly the range, azimuth and Doppler frequency of the targets, we apply the concept of Doppler focusing from \cite{bar2014sub} to our setting. Once the Fourier coefficients (\ref{coffAlligned2}) are acquired and processed, we perform Doppler focusing for a specific frequency $\nu$, that is
    \begin{eqnarray}
\label{eq:dop_reduced}
   \Phi^{\nu}_{m,q}[k] &=&\sum_{p=0}^{P-1} y_{m,q}^p[k] e^{-j2\pi \nu p \tau} \\
 &=&   \sum_{l=1}^{L} \alpha_l e^{j2\pi \beta_{mq} \vartheta_l}e^{-j\frac{2\pi }{\tau}(k+f_m \tau)\tau_l} \sum_{p=0}^{P-1} e^{j 2 \pi (f^D_l -\nu) p\tau}, \nonumber
       \end{eqnarray}
       for $-\frac{N}{2} \leq k \leq -\frac{N}{2}-1$.
Following the same argument as in \cite{bar2014sub}, it holds that
\begin{equation}
\label{eq:dop_foc}
\sum_{p=0}^{P-1} e^{j 2 \pi (f^D_l -\nu) p\tau} \cong \left\{ \begin{array}{ll} P & |f^D_l -\nu| < \frac{1}{2P\tau}, \\
0 & otherwise. \end{array} \right.
\end{equation}
Therefore, for each focused frequency $\nu$,  (\ref{eq:dop_reduced}) reduces to (\ref{coffAlligned}) and the resulting CS problem to solve is exactly (\ref{eq:model}), for $0 \leq m \leq M-1$.
We note that Doppler focusing increases the SNR by a factor a $P$, as can be seen in (\ref{eq:dop_foc}).

Algorithm \ref{algo:omp2} extends Algorithm \ref{algo:omp} to solve (\ref{eq:model2}) using Doppler focusing. Note that step 1 can be performed using fast Fourier transform (FFT). 
In the algorithm description, $\text{vec}(\mathbf{Z})$ is defined similarly to $\text{vec}(\mathbf{Y})$ in (\ref{eq:Y_def}), $\mathbf{e}_t(l)= \left[ (\mathbf{e}_t^0(l))^T \, \cdots \, (\mathbf{e}_t^{M-1}(l))^T \right]^T$
where $\mathbf{e}_t^m(l) = \text{vec}( (\mathbf{\bar{B}}^m \otimes \mathbf{A}^m)_{\Lambda_t(l,2)TN+\Lambda_t(l,1)} \left((\mathbf{\bar{F}}^m)^T_{\Lambda_t(l,3)} \right)^T)$ with $\Lambda_t(l,i)$ the $(l,i)$th element in the index set $\Lambda_t$ at the $t$th iteration, and $\mathbf{E}_t= [\mathbf{e}_t(1) \, \dots \, \mathbf{e}_t(t)]$.
Once $\mathbf{X}^p$ are recovered, the delays and azimuths are given by (\ref{eq:hat_tau}) and (\ref{eq:hat_theta}), respectively and the Dopplers are estimated as
\begin{equation}
\hat{f}_l^D = -\frac{1}{2\tau}+\frac{1}{P\tau}\Delta_L(l,3).
\end{equation}

\begin{algorithm}
\caption{OMP for simultaneous sparse 3D recovery with focusing}\label{algo:omp2} 
\begin{algorithmic}[1]
\qinput Observation matrices $\mathbf{Z}^{(m,p)}$, measurement matrices $\mathbf{A}^{(m,p)}$, $\mathbf{B}^{(m,p)}$, for all $0 \leq m \leq M-1$ and $0 \leq p \leq P-1$
\qoutput Index set $\Lambda$ containing the locations of the non zero indices of $\mathbf{X}$, estimate for sparse matrix $\mathbf{\hat{X}}$
\State Perform Doppler focusing for $0 \leq i \leq TN$ and $0 \leq j \leq TR$:
$$
\mathbf{\Phi}^{(m,\nu)}_{i,j} = \sum_{p=0}^{P-1} \mathbf{\mathbf{Y}}^{(m,p)}_{i,j} e^{j 2 \pi \nu p \tau}.
$$
\State Initialization: residual $\mathbf{R}_0^{(m,p)}=\mathbf{\Phi}^{(m,p)}$, index set $\Lambda_0=\emptyset$, $t=1$
\State Project residual onto measurement matrices for $0 \leq p \leq P-1$:
$$
\mathbf{\Psi}^p =\mathbf{A}^H \mathbf{R}^p \mathbf{B},
$$
where $\mathbf{A}$ and $\mathbf{B}$ are defined in (\ref{eq:bigA}) and (\ref{eq:bigB}), respectively, and $\mathbf{R}^p = \text{diag} \left( [ \mathbf{R}^{(0,p)}_{t-1} \, \cdots \, \mathbf{R}^{(M-1, p)}_{t-1}] \right)$ is block diagonal
\State Find the three indices $\lambda_t = [\lambda_t(1) \, \lambda_t(2) \, \lambda_t(3)]$ such that
$$
[\lambda_t(1) \quad \lambda_t(2) \quad \lambda_t(3)] = \text{ arg max}_{i,j,\nu} \left| \mathbf{\Phi}^{\nu}_{i,j} \right|
$$
\State Augment index set $\Lambda_t = \Lambda_t  \bigcup \{ \lambda_t \}$
\State Find the new signal estimate
$$
\mathbf{\hat{\alpha}}=[\hat{\alpha}_{1} \, \dots \, \hat{\alpha}_{t}]^T = ( \mathbf{E}_t^T \mathbf{E}_t )^{-1} \mathbf{E}_t^T \text{vec}(\mathbf{Z}) 
$$
\State Compute new residual
$$
\mathbf{R}_t^{(m,p)}= \mathbf{Y}^m- \sum_{l=1}^t \alpha_l e^{j2 \pi \left(-\frac{1}{2}+\frac{ \Lambda_t(l,3)}{P} \right) p} \mathbf{a}^m_{\Lambda_t(l,1)} \left( \mathbf{\bar{b}}^m_{\Lambda_t(l,2)} \right)^T$$
\State If  $t < L$, increment $t$ and return to step 2, otherwise stop
\State Estimated support set $\hat{\Lambda}= \Lambda_L$ \\
Estimated matrix $\mathbf{\hat{X}}_D$: $\left( \Lambda_L(l,2) TN+ \Lambda_L(l,1) ,  \Lambda_L(l,3)\right)$-th component of $\mathbf{\hat{X}}_D$ is given by $\hat{\alpha}_l$ for $l=1, \cdots, L$ while rest of the elements are zero
\end{algorithmic}
\end{algorithm}

\subsection{Multi-Carrier SUMMeR}

We now explain how the frequency bands left vacant can be exploited to increase the system's performance without expanding the total bandwidth of $B_{\text{tot}}=TB_h$, thus preserving assumption \textbf{A3} (\ref{eq:A3_1}) and \textbf{A5} (\ref{eq:A1}). Denote by $\gamma=T/M$ the compression ratio of the number of transmitters. In this configuration, referred to as multi-carrier SUMMeR, each transmit antenna sends $\gamma$ pulses in each PRI. Each pulse belongs to a different frequency band and are therefore mutually orthogonal, such that the total number of user bands is $M\gamma B_h=TB_h$. The $i$th pulse of the $p$th PRI is transmitted at time $i\frac{\tau}{\gamma}+p\tau$, for $0 \leq i < \gamma$ and $0 \leq p \leq P-1$. The samples are then acquired and processed as described above. Besides increasing the detection performance as we show in simulations in Section \ref{sec:sim}, this method multiplies the Doppler dynamic range by a factor of $\gamma$ with the same Doppler resolution since the CPI, equal to $P\tau$, is unchanged. Preserving the CPI allows to preserve the stationary condition on the targets, that is assumptions $\textbf{A2}$, $\textbf{A3}$ (\ref{eq:A3_1}) and $\textbf{A4}$ are still valid.

\subsection{Design Parameters}

To conclude the theoretical section, we provide some guidelines for choosing the design parameters values. From Theorems \ref{th:cond} and \ref{th:cond2}, the minimal number of channels to recover $2L$ targets is $MQ=2L$. To minimize the total number of antennas $M+Q$, we choose $M, Q \in \mathbb{N}$, such that $\sqrt{2L}-1 \leq M,Q \leq \sqrt{2L}+1$ and $MQ \geq 2L$. The number of samples per channel should then be at least $K \geq 2L/M$ with $P \geq 2L$ pulses. Obviously, these numbers are lower bounds and should be increased in the presence of noise. 
\section{Simulations}
\label{sec:sim}

In this section, we present some numerical experiments illustrating both our range-azimuth and range-azimuth-Doppler recovery approaches. We compare our method with classic MIMO processing and examine the impact of the choice of several design parameters on the detection performance.

\subsection{Preliminaries}

Throughout the experiments, the standard MIMO system is based on a virtual array, as depicted in Fig.~\ref{fig:arrays2}(a), which would be generated by $T=20$ transmit antennas and $R=20$ receive antennas, yielding an aperture $\lambda Z=6m$. The SUMMeR system is composed of $M<T$ transmitters and $Q<R$ receivers, with locations generated uniformly at random over the virtual array, as shown in Fig.~\ref{fig:arrays2}(b). We use FDMA waveforms $h_m(t)$ such that $f_m=(i_m-\frac{T}{2})B_h$, where $i_m$ are integers chosen uniformly at random in $[0,T]$, and with the following parameters: PRI $\tau=100 \mu sec$, bandwidth $B_h = 5MHz$ and carrier $f_c=10 GHz$. We consider targets from the Swerling-0 model with identical amplitudes and random phases. The received signals are corrupted with uncorrelated additive Gaussian noise (AWGN) with power spectral density $N_0$. The SNR is defined as
\begin{equation}
\text{SNR} = \frac{\frac{1}{T_p} \int_0^{T_p} |h_0(t)|^2 \mathrm{d}t}{N_0  B_h}.
\end{equation}

We consider a hit-or-miss criterion as performance metric. A ``hit" is defined as a range-azimuth estimate which is identical to the true target position up to one Nyquist bin (grid point) defined as $1/TB_h$ and $2/TR$ for the range and azimuth, respectively. In pulse-Doppler settings, a ``hit" is proclaimed if, in addition, the recovered Doppler is identical to the true frequency up to one Nyquist bin of size $1/P\tau$.

\subsection{Numerical Results}

We first consider a sparse target scene with $L=7$ targets including a couple of targets with close ranges and another couple close in azimuth, both up to one grid point. We use $M=10$ transmit antennas and $Q=10$ receive antennas and employ $K=250$ samples per channel instead of $N= B_h \tau = 500$, which corresponds to only $12.5 \%$ of the total number of Nyquist rate samples from the original array.
The SNR is set to $0$dB. Figure \ref{fig:sim1} (left pane) shows the sparse target scene on a range-azimuth map, where each real target is displayed with its estimated location. In the right pane, we demonstrate range-azimuth-Doppler recovery and show the location and velocity of $L=6$ targets, including a couple of targets with close ranges, a couple with close azimuths and another couple with close velocities. Here, the range and azimuth are converted to 2-dimensional $x$ and $y$ locations. The SNR is set to $-10$dB.

\begin{figure}[!h]
\begin{center}
\includegraphics[width=0.2\textwidth]{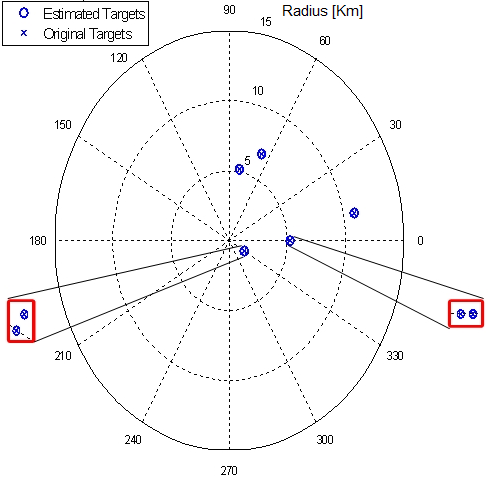}
\includegraphics[width=0.2\textwidth]{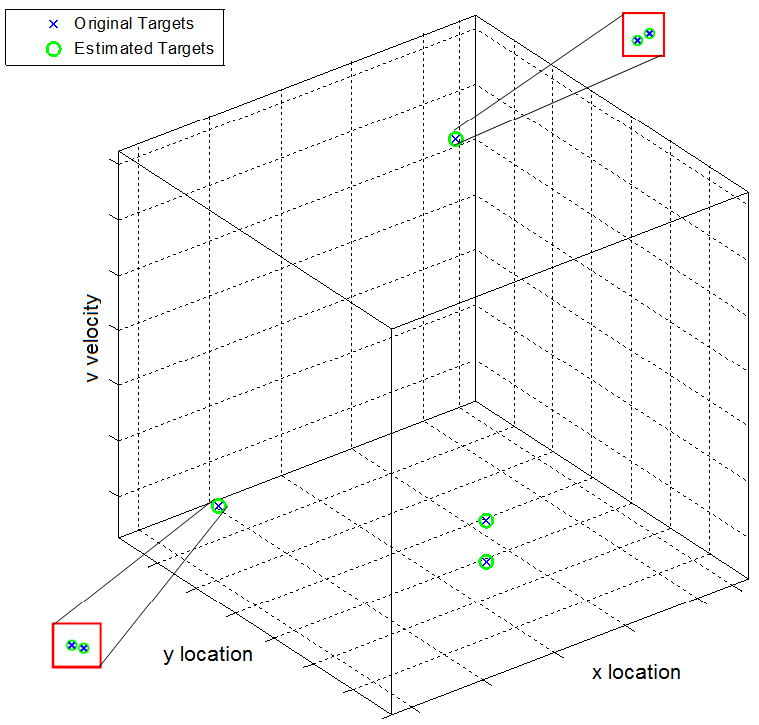}
\caption{Range-azimuth recovery for $L=7$ targets and SNR=$0$dB (left), range-azimuth-Doppler recovery for $L=6$ targets and SNR=$-10$dB (right).}
\label{fig:sim1}
\end{center}
\end{figure}



Next, we investigate the performance of our azimuth-range-Doppler recovery scheme with respect to SNR for different numbers of samples $K$ per channel. We use the same array as described above, with spatial compression of $25\%$, where each transmitter sends $P=10$ pulses,. We consider $L=10$ targets whose locations are generated uniformly at random. Each experiment is repeated over $100$ realizations. Figure~\ref{fig:sim_timeP} presents the range-azimuth-Doppler recovery performance with respect to SNR. The configuration with $K=500$ corresponds to samples obtained at the Nyquist rate. The configuration with $K=125$ results in $1:4$ time compression along with half the number of transmitters and receivers, yielding only $6.25 \%$ of the total number of Nyquist rate samples from the original array. We observe a shift of $3$dB between the consecutive graphs, since the compression between them yield a decrease of half the system power. Similar results are observed for spatial compression.

\begin{figure}[!h]
\begin{center}
\includegraphics[width=0.4\textwidth]{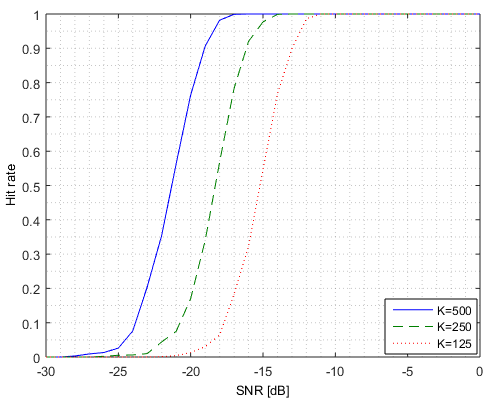}
\caption{Range-azimuth-Doppler recovery performance with time compression.}
\label{fig:sim_timeP}
\end{center}
\end{figure}

%

%

Next, we compare our proposed range-azimuth recovery approach to the classic MIMO processing presented in Section~\ref{sec:classic}. For the classic method, we use bandlimited Gaussian pulses with bandwidth $B_h=100$MHz that are equivalent to the CDMA approach, and we neglect the wideband effects only for CDMA, giving it an advantage. 
Here, the SNR is defined with respect to the CDMA classic processing for the two methods to be comparable, namely 
\begin{equation} \label{eq:snr2}
\text{SNR}=\frac{\frac{M}{T_p} \int_0^{T_p} |h_0(t)|^2 \mathrm{d}t}{N_0  B_{\text{tot}}}.
\end{equation}
We consider $L=2$ close targets in the spatial domain, up to one azimuth grid point. In each simulation, the pair of targets is generated at random. We consider two regimes: the first does not involve any compression, namely $M=T=20$, $Q=R=20$ and $K=N=500$, while in the second, spatial compression is considered with $M=Q=10$. Figure \ref{fig:comp_spat} presents the range-azimuth recovery performance with respect to SNR for both approaches and both regimes. 
We first observe that without compression both our FDMA approach and the conventional CDMA method are equivalent in terms of performance. Here, the wideband effects were neglected only for the classic CDMA method, which is thus given an advantage. In our settings, the delay $\eta_{mq}\vartheta_l$ due to the array geometry cannot be neglected, as in classic CDMA, since the array aperture $\lambda Z=6m$ is only 4 times the range resolution $1/TB_h=1.25m$ so that (\ref{eq:A1}) does not hold.
With wideband effects applied on CDMA, FDMA associated with our processing would then outperform the latter. In the case of $25 \%$ spatial compression in Fig.~\ref{fig:comp_spat}, the azimuth resolution of the classic approach is divided by $4$, whereas the resolution of our approach remains unchanged. Therefore, even at high SNRs, two close targets in the spatial domain cannot be resolved since one can mask the other. Note that in low SNR, the classic processing seems to yield higher performance than ours. This stems from the fact that if the two close targets are in phase, they can produce constructive signals and the classic approach will in fact detect the sum of both. Similar observations can be made under time compression. 

\begin{figure}[!h]
\begin{center}
\includegraphics[width=0.4\textwidth]{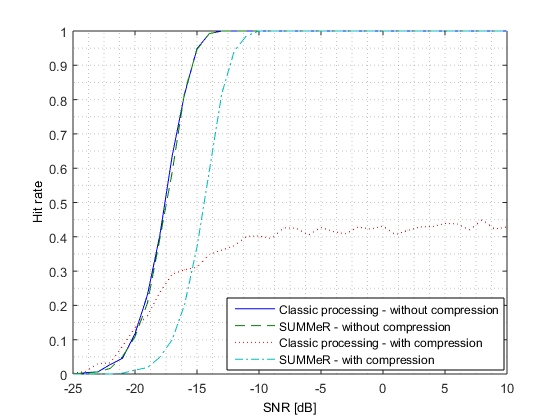}
\caption{Range-azimuth recovery performance of our proposed method vs. classic processing (spatial, or azimuth, resolution).}
\label{fig:comp_spat}
\end{center}
\end{figure}


Last, we illustrate the increased detection performance achieved by the multi-carrier SUMMeR method. We consider two regimes: the first does not involve any compression, while in the second, spatial compression is considered with $M=Q=10$. In the classic and SUMMeR system, $P=10$ pulses are transmitted by each transmit antenna. In multi-carrier SUMMeR, we have $\gamma=2$, leading to $2P=20$ pulses per transmitter. For all configurations, we consider $L=5$ targets. In Fig.~\ref{fig:duplex}, we observe that the multi-carrier approach with spatial compression achieves the same performance as the original SUMMeR and the classic processing with no compression. While using the SNR definition of (\ref{eq:snr2}), the number of transmitters does not affect the detection performance, since we have a system with fixed power. The reduction of the number of receivers decreases the performance by 3dB, which are compensated by the extra transmitted pulses. This shows that the same detection performed can be achieved by keeping the total transmission bandwidth while reducing the number of antennas. In fact, if the transmitters' reduction factor is greater than the receivers', then the detection performance of the multi-carrier SUMMeR system will be higher than that achieved by conventional processing without any compression.
\begin{figure}[!h]
\begin{center}

\includegraphics[width=0.4\textwidth]{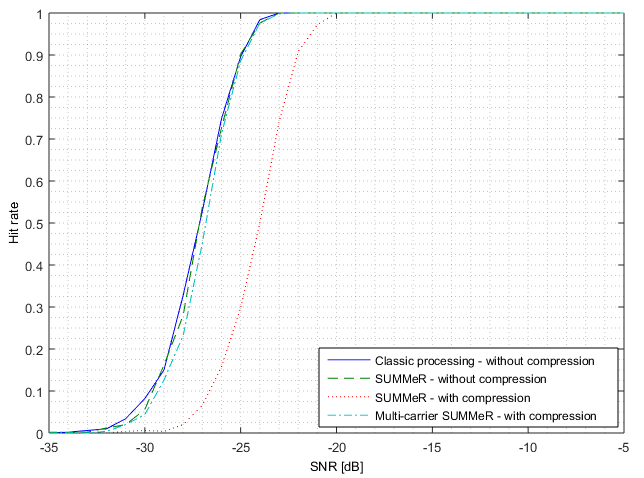}
\caption{Multi-carrier SUMMeR with spatial compression vs. classic processing and SUMMeR.}
\label{fig:duplex}
\end{center}
\end{figure}

%

\subsection{Choice of Parameters}
\label{sec:sim_param}

In the above experiments, design parameters such as antennas' locations, transmissions' carrier frequencies and Fourier coefficients were chosen at random. In \cite{cohen2016nyquist}, the impact of the joint choice of locations and carriers on range-azimuth coupling is explored. It is explained that in order to overcome the ambiguity issue, either the antenna locations or the carrier frequencies, or both, should not lie on a specific grid. It was heuristically found that a configuration with random antennas' locations with carriers on a grid provides better results than random carriers with a ULA structure.

We investigate the impact of the choice of Fourier coefficients, antennas' locations and carrier frequencies on the side lobes on the range and azimuth domains. The azimuth peak sidelobe level is governed by the antennas' locations and carriers, which are jointly chosen as to reduce the coherence of $\bf B$, defined in (\ref{eq:bigB}). To reduce the range peak sidelobe level, we manually choose Fourier coefficients and carriers that jointly lead to low coherence of $\mathbf{A}$, defined in (\ref{eq:bigA}). This procedure is complex due to the inter-connection between both parameters. The peaks sidelobe levels obtained by following these guidelines are shown to be lower than those achieved with typical random parameter selection in Figs.~\ref{fig:amb3} and \ref{fig:amb4}.
\begin{figure}
\begin{center}
\includegraphics[width=0.3\textwidth]{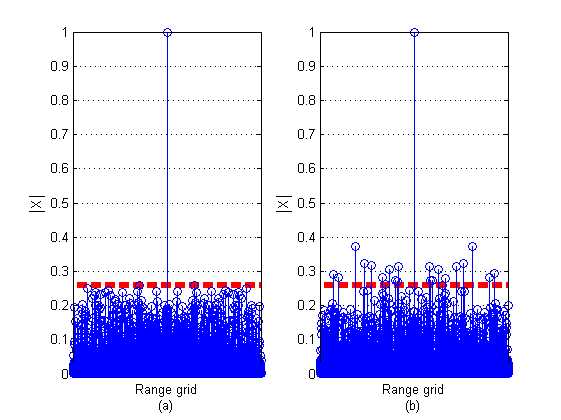}
\caption{Range-azimuth map in noiseless settings along range axis for choice of Fourier coefficients based on low coherence (left) and typical random choice (right), for $K=50$, $M=Q=20$ and $L=1$. The red dotted line indicates the peak sidelobe level for this target in the range domain.}
\label{fig:amb3}
\end{center}
\end{figure}
\begin{figure}
\begin{center}
\includegraphics[width=0.3\textwidth]{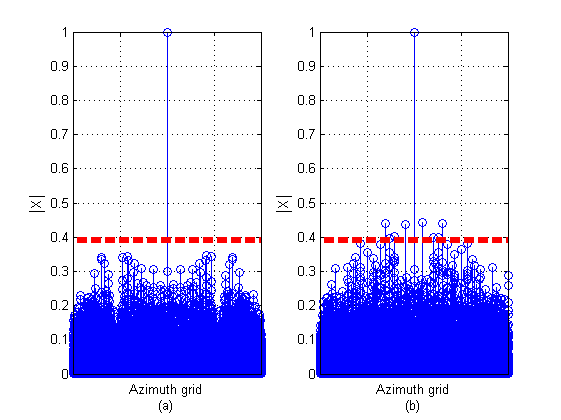}
\caption{Range-azimuth map in noiseless settings along range axis for choice of antennas' locations based on low coherence (left) and typical random choice (right), for $K=50$, $M=Q=20$ and $L=1$. The red dotted line indicates the peak sidelobe level for this target in the range domain.}
\label{fig:amb4}
\end{center}
\end{figure}

\section{Conclusion}
In this work, we presented the SUMMeR system, a sub-Nyquist MIMO radar sampling and recovery method, which exploits the concept of Xampling and Doppler focusing. This system breaks both the links between sampling rate and time resolution, and number of antennas and spatial resolution. We derived necessary conditions for range, azimuth and Doppler recovery, both with and without grid assumptions, that translate into minimal numbers of channels and samples per receiver for perfect recovery in noiseless settings. 
We compared our method with the classic Nyquist MIMO processing and showed that both the time and spatial resolution of our approach is preserved under time and spatial compression, as opposed to the traditional method. Furthermore, we proposed an enhanced version of SUMMeR, that exploits the frequency bands left vacant due to spatial compression, to recover the lost detection performance from this compression. We next investigated the impact of designs parameters, such as antennas' locations, carrier frequencies and chosen Fourier coefficients, on the detection performance. Last, our processing provides a solution to FDMA main drawbacks, range-azimuth coupling and range resolution limited to a single waveform's bandwidth, achieving the same performance as CDMA, for which wideband effects were neglected. In the presence of these, FDMA associated with our processing would outperform CDMA.


\appendices
\section{Proof Lemma \ref{lemma:Cspark}} \label{sec:app}
\textbf{Lemma \ref{lemma:Cspark}.}
Let $\mathbf{A}^m \in \mathbb{C}^{K,N}$ and $\mathbf{B}^m \in \mathbb{C}^{Q,R}$, for $0 \leq m \leq M-1$ with $K \leq N$ and $Q \leq R$. Denote $\mathbf{A} = [ \mathbf{A}^{0^T} \, \mathbf{A}^{1^T} \, \cdots \, \mathbf{A}^{(M-1)^T} ]^T$ and $\mathbf{B} = [ \mathbf{B}^{0^T} \, \mathbf{B}^{1^T} \, \cdots \, \mathbf{B}^{(M-1)^T}]^T$. Let
\begin{equation}
\mathbf{C} = \left[ 
\begin{array}{c} \mathbf{B}^0 \otimes \mathbf{A}^0 \\ \mathbf{B}^1 \otimes \mathbf{A}^1 \\ \vdots \\ \mathbf{B}^{M-1} \otimes \mathbf{A}^{M-1}
\end{array} \right].
\end{equation}
Then,
\begin{equation}
\text{spark}(\mathbf{C}) = \min \{ \text{spark}(\mathbf{A}), \text{spark}(\mathbf{B}) \}.
\end{equation}
\begin{proof}
The $MKQ \times NR$ matrix $\bf C$ can be expressed more explicitly as
\begin{equation}
\mathbf{C} = \left[ 
\begin{array}{ccc} b^0_{11} \mathbf{A}^0 & \cdots & b^0_{1R} \mathbf{A}^0 \\ \vdots & \ddots & \vdots \\ b^0_{Q1} \mathbf{A}^0 & \cdots & b^0_{QR} \mathbf{A}^0 \\ \vdots & \vdots & \vdots \\ b^{M-1}_{11} \mathbf{A}^{M-1} & \cdots & b^{M-1}_{1R} \mathbf{A}^{M-1} \\ \vdots & \ddots & \vdots \\ b^{M-1}_{Q1} \mathbf{A}^{M-1} & \cdots & b^{M-1}_{QR} \mathbf{A}^{M-1}
\end{array} \right].
\end{equation}

We first show that $\text{spark}(\mathbf{C}) \leq \min \{ \text{spark}(\mathbf{A}), \text{spark}(\mathbf{B}) \}$.
By definition of $\text{spark}(\mathbf{A})$, there exists a vector $\mathbf{x}_A \in \mathbb{C}^{N}$ with $||\mathbf{x}_A||_0 = \text{spark}(\mathbf{A})$ such that $\mathbf{A} \mathbf{x}_A =0$. Equivalently, $\mathbf{A}^m \mathbf{x}_A =0$, for all $0 \leq m \leq M-1$. Let $\mathbf{y}_A = \left[ \mathbf{x}_A^T \, 0 \, \cdots \, 0 \right]^T \in \mathbb{C}^{RN}$. Then, it holds that $\mathbf{C}\mathbf{y}_A =0$ with $||\mathbf{y}_A||_0 = ||\mathbf{x}_A||_0=\text{spark}(\mathbf{A})$. Thus, $\text{spark}(\mathbf{C}) \leq \text{spark}(\mathbf{A})$.

We now use the fact that there exists permutation matrices $\mathbf{\Pi}_1$ and $\mathbf{\Pi}_2$ such that $\mathbf{B}^m \otimes \mathbf{A}^m = \mathbf{\Pi}_1 \left(\mathbf{A}^m \otimes \mathbf{B}^m \right) \mathbf{\Pi}_2$ \cite{jokar2009sparse}, for all $0 \leq m \leq M-1$. By definition of $\text{spark}(\mathbf{B})$, there is a vector $\mathbf{x}_B \in \mathbb{C}^{R}$ with $||\mathbf{x}_B||_0 = \text{spark}(\mathbf{B})$ such that $\mathbf{B} \mathbf{x}_B =0$. Let $\mathbf{\tilde{y}}_B = \left[ \mathbf{x}_B^T \, 0 \, \cdots \, 0 \right]^T \in \mathbb{C}^{NR}$ and $\mathbf{y}_B = \mathbf{\Pi}_2^{-1} \mathbf{\tilde{y}}_B$. Rewriting $\bf C$ as
\begin{equation}
\mathbf{C} = \left[ 
\begin{array}{c} \mathbf{\Pi}_1 \left( \mathbf{A}^0 \otimes \mathbf{B}^0 \right) \mathbf{\Pi}_2 \\ \mathbf{\Pi}_1 \left( \mathbf{A}^1 \otimes \mathbf{B}^1 \right) \mathbf{\Pi}_2\\ \vdots \\ \mathbf{\Pi}_1 \left( \mathbf{A}^{M-1} \otimes \mathbf{B}^{M-1} \right) \mathbf{\Pi}_2
\end{array} \right],
\end{equation}
we have $\mathbf{C}\mathbf{y}_B=0$ with $||\mathbf{y}_B||_0 = ||\mathbf{x}_B||_0=\text{spark}(\mathbf{B})$. Therefore, $\text{spark}(\mathbf{C}) \leq \text{spark}(\mathbf{B})$.

We now show that $\text{spark}(\mathbf{C}) \geq \min \{ \text{spark}(\mathbf{A}), \text{spark}(\mathbf{B}) \}$. Assume first that $\text{spark}(\mathbf{A}) \geq \text{spark}(\mathbf{B})$. Then, we need to show that $\text{spark}(\mathbf{C}) \geq \text{spark}(\mathbf{B})$. 
Indeed, every column of $\bf C$ has the form
\begin{equation}
\mathbf{c}_{w_j} = \left[ \begin{array}{c}
\mathbf{b}_{v_j}^0 \otimes \mathbf{a}_{u_j}^0 \\ \mathbf{b}_{v_j}^1 \otimes \mathbf{a}_{u_j}^1 \\ \vdots \\ \mathbf{b}_{v_j}^{M-1} \otimes \mathbf{a}_{u_j}^{M-1}
\end{array} \right],
\end{equation}
for $0 \leq w_j \leq NR-1$, $0 \leq u_j \leq N-1$, and $0 \leq v_j \leq R-1$.
Suppose by contradiction that
\begin{equation}
\text{spark}(\mathbf{C}) = \ell < \text{spark}(\mathbf{B}). 
\end{equation}
In particular, this implies that any set of $\ell$ columns of $\mathbf{B}$ is linearly independent, while there exist scalars $\lambda_1, \dots, \lambda_{\ell}$ not all $0$ and indices $u_1, \dots, u_{\ell}$ and $v_1, \dots, v_{\ell}$ where $v_i \neq v_j$ for all $i \neq j$ such that
\begin{equation}
\label{eq:lemma_sum}
\sum_{j=1}^{\ell} \lambda_j \mathbf{c}_{w_j} = \sum_{j=1}^{\ell} \left[ \begin{array}{c}
\left(\lambda_j \mathbf{b}_{v_j}^0 \right) \otimes \mathbf{a}_{u_j}^0 \\
\left(\lambda_j \mathbf{b}_{v_j}^1 \right) \otimes \mathbf{a}_{u_j}^1 \\ \vdots \\
\left(\lambda_j \mathbf{b}_{v_j}^{M-1} \right) \otimes \mathbf{a}_{u_j}^{M-1}
\end{array} \right] = \mathbf{0}.
\end{equation}
In (\ref{eq:lemma_sum}), each index $u_j$ can appear multiple times. Without loss of generality, we assume that the indices $u_j$ are numbered in increasing order, so that
\begin{equation}
\underbrace{u_1= \dots = u_{k_1}}_{g_1} < \dots < \underbrace{u_{k_t+1}= \dots = u_{k_t}}_{g_t},
\end{equation}
with $1 \leq t \leq \ell$.
Therefore, we have
\begin{equation} 
\sum\limits_{i=1}^t\left[ \begin{array}{c}
\left( \sum\limits_{j=k_{i-1}+1}^{k_i} \lambda_j \mathbf{b}_{v_j}^0 \right) \otimes \mathbf{a}_{g_i}^0 \\
\left( \sum\limits_{j=k_{i-1}+1}^{k_i} \lambda_j \mathbf{b}_{v_j}^1 \right) \otimes \mathbf{a}_{g_i}^1 \\ \vdots \\
\left( \sum\limits_{j=k_{i-1}+1}^{k_i} \lambda_j \mathbf{b}_{v_j}^{M-1} \right) \otimes \mathbf{a}_{g_i}^{M-1}
\end{array} \right] 
= 0,
\end{equation}
where $k_0=0$ and $k_t = \ell$. Since $\text{spark}(\mathbf{A}) > \ell$, the vectors $\mathbf{a}_{g_1}, \dots, \mathbf{a}_{g_{\ell}}$ are linearly independent. It follows that
\begin{equation}
\label{eq:lemma_last}
\sum\limits_{j=k_{i-1}+1}^{k_i} \lambda_j \left[ \begin{array}c
\mathbf{b}_{v_j}^0 \\ \mathbf{b}_{v_j}^1 \\ \vdots \\ \mathbf{b}_{v_j}^{M-1} \end{array}
\right] =0,
\end{equation}
for $1 \leq i \leq t$.
Since the sum in (\ref{eq:lemma_last}) is over at most $\ell$ columns of $\bf B$, this contradicts the assumption that $\text{spark}(\mathbf{B}) > \ell$.

Finally, assume $\text{spark}(\mathbf{B}) \geq \text{spark}(\mathbf{A})$. We then need to show that $\text{spark}(\mathbf{C}) \geq \text{spark}(\mathbf{A})$. This can be proved similarly to the previous case by writing the columns of $\bf C$ as
\begin{equation}
\mathbf{c}_{w_j} = \left[ \begin{array}{c}
\mathbf{\Pi}_1 \left( \mathbf{b}_{v_j}^0 \otimes \mathbf{a}_{u_j}^0 \right) \\  \mathbf{\Pi}_1 \left( \mathbf{b}_{v_j}^1 \otimes \mathbf{a}_{u_j}^1 \right) \\ \vdots \\ \mathbf{\Pi}_1 \left( \mathbf{b}_{v_j}^{M-1} \otimes \mathbf{a}_{u_j}^{M-1} \right)
\end{array} \right],
\end{equation}
for $0 \leq w_j \leq NR-1$, $0 \leq u_j \leq N-1$, $0 \leq v_j \leq R-1$ and where $\mathbf{\Pi}_1$ is an appropriate permutation matrix.

\end{proof}

\small  
\bibliographystyle{IEEEtran}
\bibliography{Bib}

\end{document}